\documentclass[11pt]{article}

\usepackage[utf8]{inputenc}
\usepackage{hyperref}
\usepackage{amsthm, amssymb}
\usepackage{amsmath}
\usepackage{fullpage}
\usepackage{graphicx}
\usepackage{cite}
\usepackage[colorinlistoftodos]{todonotes}
\usepackage{tikz,pgfplots}
\usetikzlibrary{decorations}
\usetikzlibrary{patterns}
\usetikzlibrary{backgrounds}
\usetikzlibrary{scopes}
\usetikzlibrary{arrows, automata, positioning}
\usepackage{xspace}
\usepackage{bookmark}
\usepackage{caption}
\usepackage{subcaption}
\usepackage{algpseudocode}
\usepackage{algorithm}
\renewcommand{\algorithmicrequire}{\textbf{Input:}}

\usepackage{microtype}

\newtheorem{theorem}{Theorem}
\newtheorem{lemma}[theorem]{Lemma}
\newtheorem{definition}[theorem]{Definition}
\newtheorem*{remark}{Remark}
\newtheorem{observation}[theorem]{Observation}

\newcommand{\e}{\varepsilon}
\newcommand{\J}{\mathcal{J}}
\newcommand{\M}{\mathcal{M}}
\newcommand{\Sc}{\mathcal{S}}
\newcommand{\LPT}{\textnormal{LPT}}
\newcommand{\OPT}{\textnormal{OPT}}
\newcommand{\Alg}{\mathcal{A}}
\newcommand{\UB}{\textnormal{\footnotesize{UB}}}
\newcommand{\tp}{\tilde{p}}
\newcommand{\tP}{\tilde{P}}
\newcommand{\load}{\textnormal{load}}

\newcommand{\lmin}{\ell_{\min}}
\newcommand{\Me}{\mathcal{M}^{=}}
\newcommand{\Mne}{\mathcal{M}^{\neq}}
\newcommand{\Je}{\mathcal{J}^{=}}

\newcommand{\jvcom}[1]{}
\newcommand{\jscom}[1]{}
\newcommand{\wgcom}[1]{}

\algrenewcommand\algorithmiccomment[2][\normalsize]{{#1\hfill\(\triangleright\) \texttt{#2}}}
	
\title{Symmetry exploitation for Online Machine Covering with Bounded Migration\footnote{This work was partially supported by FONDECYT project 11140579, FONDECYT project 11130266, and by Nucleo Milenio Informaci\'on y Coordinaci\'on en Redes ICM/FIC RC130003.}}

\author{Waldo G\'alvez\footnote{IDSIA, USI-SUPSI, Lugano, Switzerland. E-mail: \texttt{waldo@idsia.ch}} \and Jos\'e A. Soto\footnote{Departamento de Ingeniería Matemática \& CMM, Universidad de Chile, Santiago, Chile. E-mail: \texttt{jsoto@dim.uchile.cl}} \and Jos\'e Verschae\footnote{Facultad de Matemáticas \& Escuela de Ingeniería, Pontificia Universidad Católica de Chile, Santiago, Chile. E-mail: \texttt{jverschae@uc.cl}}}

\begin{document}
	\maketitle
	\begin{abstract}
	
	Online models that allow recourse are highly effective in situations where classical models are too pessimistic. One such problem is the online machine covering problem on identical machines. 
	In this setting, jobs arrive one by one and must be assigned to machines with the objective of maximizing the minimum machine load. 
	When a job arrives, we are allowed to reassign some jobs as long as their total size is (at most) proportional to the processing time of the arriving job. The proportionality constant is called the \emph{migration factor} of the algorithm.
	
	Using a rounding procedure with useful structural properties for online packing and covering problems, we design first a simple $(1.7 + \e)$-competitive algorithm using a migration factor of $O(1/\e)$ which maintains at every arrival a locally optimal solution with respect to the Jump neighborhood. After that, we present as our main contribution a more involved $(4/3+\e)$-competitive algorithm using a migration factor of $\tilde{O}(1/\e^3)$. At every arrival, we run an adaptation of the \emph{Largest Processing Time first} ($\LPT$) algorithm. Since the new job can cause a complete change of the assignment of smaller jobs in both cases, a low migration factor is achieved by carefully exploiting the highly symmetric structure obtained by the rounding procedure.
\end{abstract}

\section{Introduction} 
We consider a fundamental load balancing problem where $n$ jobs need to be assigned to $m$ identical parallel machines. Each job $j$ is fully characterized by a non-negative processing time $p_j$. Given an assignment of jobs, the load of a machine is the sum of the processing times of jobs assigned to it. The \emph{machine covering problem} asks for an assignment of jobs to machines maximizing the load of the least loaded machine.

This problem is well known to be strongly $\text{NP}$-hard and allows for a polynomial-time approximation scheme (PTAS)~\cite{W97}. A well studied algorithm for this problem is the \emph{Largest Processing Time First} rule ($\LPT$), that sorts the jobs non-increasingly and assigns them iteratively to the least loaded machine. Deuermeyer et al. \cite{DFL82} show that $\LPT$ is a $\frac{4}{3}$-approximation and that this factor is asymptotically tight; later, Csirik et al. \cite{CKW92} refine the analysis giving a tight bound for each~$m$.

In the online setting jobs arrive one after another, and at the moment of an arrival, we must decide on a machine to assign the arriving job. This natural problem does not admit a constant competitive ratio. Deterministically, the best possible competitive ratio is~$m$~\cite{W97}, while randomization allows for a $\tilde{O}(\sqrt{m})$-competitive algorithm, which is the best possible up to logarithmic factors~\cite{AE98}.

\paragraph*{\textbf{Dynamic model}.} The previous negative facts motivate the study of a relaxed online scenario with \emph{bounded migration}. Unlike the classic online model, when a new job $j$ arrives we are allowed to reassign other jobs. More precisely, given a constant $\beta>0$, we can migrate jobs whose total size is upper bounded by $\beta p_j$. The value $\beta$ is called the \emph{migration factor} and it accounts for the robustness of the computed solutions. In one extreme, we can model the usual online framework by setting $\beta=0$. In the other extreme, setting $\beta=\infty$ allows to compute the optimal offline solution in each iteration. Our main interest is to understand the exact trade-off between the migration factor $\beta$ and the competitiveness of our algorithms. Besides being a natural problem with an interesting theoretical motivation, its original purpose was to find good algorithms for a problem in the context of Storage Area Networks~(SAN)~\cite{SSS09}. 

\paragraph*{\textbf{Local search and migration.}} The local search method has been extensively used to tackle different hard combinatorial problems, and it is closely related to online algorithms where recourse is allowed. This comes from the fact that simple local search neighborhoods allow to get considerably improved solutions while having accurate control over the recourse actions needed, and in some cases even a bounded number of local moves leads to substantially improved solutions (see \cite{MSVW16,GGK16,Lacki2015} for examples in network design problems).

\paragraph*{\textbf{Related Work.}}

Sanders et al.~\cite{SSS09} develop online algorithms for load balancing problems in the migration framework. For the makespan minimization objective, where the aim is to minimize the maximum load, they give a $(1+\e)$-competitive algorithm with migration factor $2^{\tilde{O}(1/\e)}$. A mayor open problem in this area is to determine whether a migration factor of $\text{poly}(1/\e)$ is achievable. 

The landscape for the machine covering problem is somewhat different. Sanders et al.~\cite{SSS09} give a $2$-competitive algorithm with migration factor $1$, and this is until now the best competitive ratio known for any algorithm with constant migration factor. On the negative side, Skutella and Verschae \cite{SV10} show that it is not possible to maintain arbitrarily near optimal solutions using a constant migration factor, giving a lower bound of $20/19$ for the best competitive ratio achievable in that case. The lower bound is based on an instance where arriving jobs are very small, which do not allow to migrate any other job. This motivated the study of an amortized version, called  \emph{reassignment cost model}, where they develop a $(1+\e)$-competitive algorithm using a constant reassignment factor. They also show that if all arriving jobs are larger than $\e\cdot\OPT$, then there is a $(1+\e)$-competitive algorithm with constant migration factor. 

Similar migration models have been studied for other packing and covering problems. For example, Epstein \& Levin \cite{EL09} design a $(1+\e)$-competitive algorithm for the online bin packing problem using a migration factor of $2^{\tilde{O}(1/\e^2)}$, which was improved later by Jansen \& Klein~\cite{JK13} to $\text{poly}(1/\e)$ migration factor, and then further refined by Berndt~ et al.~\cite{BJK15}. Also, for makespan minimization with preemption and other objectives, Epstein \& Levin \cite{EL14} design a best-possible online algorithm using a migration factor of~$\left(1-\frac{1}{m}\right)$.

Regarding local search applied to load balancing problems, many neighborhoods have been studied such as \emph{Jump}, \emph{Swap}, \emph{Push} and \emph{Lexicographical Jump} in the context of makespan minimization on related machines \cite{SV07}, makespan minimization on restricted parallel machines \cite{RRSV10}, and also multi-exchange neighborhoods for makespan minimization on identical parallel machines \cite{FNS04}. For the case of machine covering, Chen et al.~\cite{CEKvS13} study the Jump neighborhood in a game-theoretical context, proving that every locally optimal solution is $1.7$-approximate and that this factor is tight.

\paragraph*{\textbf{Our Contribution.}} Our main result is a $(4/3+\e)$-competitive algorithm using $\text{poly}(1/\e)$ migration factor. This is achieved by running a carefully crafted version of $\LPT$ at the arrival of each new job. We would like to stress that, even though $\LPT$ is a simple and very well studied algorithm in the offline context, directly running this algorithm in each time step in the online context yields an unbounded migration factor; see Figure~\ref{fig:LPT_no_red} for an illustrative example and Lemma~\ref{lm:nonConstantMigration} in Appendix~\ref{app:nonConstant} for a proof. 

\begin{figure}
	\centering
	\captionsetup[subfigure]{justification=centering}
	\begin{subfigure}[b]{.5\textwidth}
		\centering
\begin{tikzpicture}[xscale=0.5,yscale=0.4]


\draw (7.5,-0.5) rectangle (12.5,4.5);
\draw (8,2) node {$1$};
\draw (9,2) node {$2$};
\draw (10,2) node {$3$};
\draw (11,2) node {$4$};
\draw (12,2) node {$5$};
\draw (12.5,-0.5) rectangle (13.5,2.75);
\draw (13,1.125) node {$6$};
\draw (13.5,-0.5) rectangle (14.5,2.5);
\draw (14,1) node {$7$};
\draw (14.5,-0.5) rectangle (15.5,2.25);
\draw (15,0.875) node {$8$};
\draw (15.5,-0.5) rectangle (16.5,2);
\draw (16,0.75) node {$9$};
\draw (12.5,2.75) rectangle (13.5,4.25);
\draw (13,3.5) node {${13}$};
\draw (13.5,2.5) rectangle (14.5,4.25);
\draw (14,3.375) node {${12}$};
\draw (14.5,2.25) rectangle (15.5,4.25);
\draw (15,3.25) node {${11}$};
\draw (15.5,2) rectangle (16.5,4.25);
\draw (16,3.125) node {${10}$};
\draw (12.5,4.25) rectangle (13.5,5.75);
\draw (13,5) node {${14}$};
\draw (13.5,4.25) rectangle (14.5,5.75);
\draw (14,5) node {${15}$};
\draw (14.5,4.25) rectangle (15.5,5.75);
\draw (15,5) node {${16}$};
\draw (15.5,4.25) rectangle (16.5,5.75);
\draw (16,5) node {${17}$};

\draw (17,-0.5) rectangle (18,3);
\draw (17.5,1.25) node {${j^*}$};


\draw (7.5,7) -- (7.5,-0.5) -- (16.5,-0.5) -- (16.5,7) ;
\draw (8.5,7) -- (8.5,-0.5);
\draw (9.5,7) -- (9.5,-0.5);
\draw (10.5,7) -- (10.5,-0.5);
\draw (11.5,7) -- (11.5,-0.5);
\draw (12.5,7) -- (12.5,-0.5);
\draw (13.5,7) -- (13.5,-0.5);
\draw (14.5,7) -- (14.5,-0.5);
\draw (15.5,7) -- (15.5,-0.5);

\end{tikzpicture}
		\caption{$\LPT$ for the original instance \\and arriving job $j^*$.}
		\label{fig:LPT_no_red1}
	\end{subfigure}%
	\begin{subfigure}[b]{.5\textwidth}
		\centering
\begin{tikzpicture}[xscale=0.5,yscale=0.4]


\draw (7.5,-0.5) rectangle (12.5,4.5);
\draw (8,2) node {$1$};
\draw (9,2) node {$2$};
\draw (10,2) node {$3$};
\draw (11,2) node {$4$};
\draw (12,2) node {$5$};
\draw (12.5,-0.5) rectangle (13.5,3);
\draw (13,1.25) node {${j^*}$};
\draw (13.5,-0.5) rectangle (14.5,2.75);
\draw (14,1.125) node {$6$};
\draw (14.5,-0.5) rectangle (15.5,2.5);
\draw (15,1) node {$7$};
\draw (15.5,-0.5) rectangle (16.5,2.25);
\draw (16,0.875) node {$8$};
\draw[draw=black,ultra thick] (12.5,3) rectangle (13.5,4.75);
\draw (13,3.875) node {${12}$};
\draw[draw=black,ultra thick] (13.5,2.75) rectangle (14.5,4.75);
\draw (14,3.75) node {${11}$};
\draw[draw=black,ultra thick] (14.5,2.5) rectangle (15.5,4.75);
\draw (15,3.625) node {${10}$};
\draw[draw=black,ultra thick] (15.5,2.25) rectangle (16.5,4.75);
\draw (16,3.5) node {${9}$};
\draw[draw=black,ultra thick]  (7.5,4.5) rectangle (12.5,6);
\draw (8,5.25) node {${17}$};
\draw (9,5.25) node {${16}$};
\draw (10,5.25) node {${15}$};
\draw (11,5.25) node {${14}$};
\draw (12,5.25) node {${13}$};
\draw[draw=black,ultra thick]  (8.5,4.5) -- (8.5,6);
\draw[draw=black,ultra thick]  (9.5,4.5) -- (9.5,6);
\draw[draw=black,ultra thick]  (10.5,4.5) -- (10.5,6);
\draw[draw=black,ultra thick]  (11.5,4.5) -- (11.5,6);


\draw (7.5,7) -- (7.5,-0.5) -- (16.5,-0.5) -- (16.5,7) ;
\draw (8.5,7) -- (8.5,-0.5);
\draw (9.5,7) -- (9.5,-0.5);
\draw (10.5,7) -- (10.5,-0.5);
\draw (11.5,7) -- (11.5,-0.5);
\draw (12.5,7) -- (12.5,-0.5);
\draw (13.5,7) -- (13.5,-0.5);
\draw (14.5,7) -- (14.5,-0.5);
\draw (15.5,7) -- (15.5,-0.5);

\end{tikzpicture}
		\caption{$\LPT$ for the new instance. Thick items correspond to migrated jobs.}
		\label{fig:LPT_no_red2}
	\end{subfigure}
	\caption{$\Omega(m)$ migration factor needed to maintain $\LPT$ at the arrival of $j^*$.}
	\label{fig:LPT_no_red}
\end{figure}
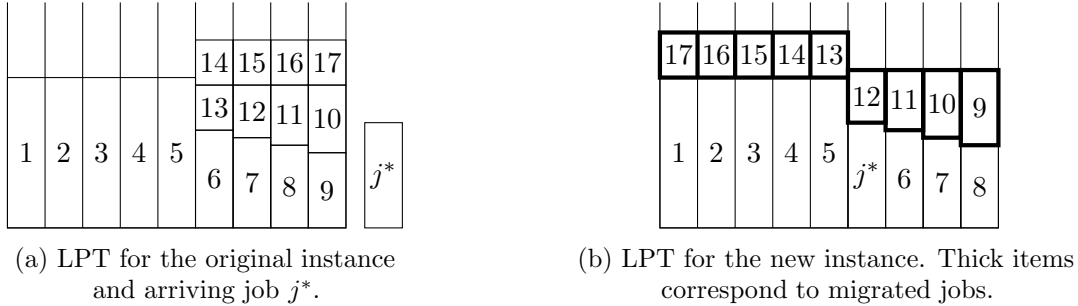

To overcome this barrier, we first adapt a less standard rounding procedure to the online framework. Roughly speaking, the rounding reduces the possible number of sizes of jobs larger than $\Omega(\e\OPT)$ (where $\OPT$ is the offline optimum value) to $\tilde{O}(1/\e)$ many numbers, and furthermore these values are multiples of a common number $g\in \Theta(\e^2\OPT)$. 
This implies that the number of possible loads for machines having only big jobs is constant since they are multiples of $g$ as well. Unlike known techniques used in previous work that yield similar results (see e.g.~\cite{JKV16}), our rounding is well suited for online algorithms and helps simplifying the analysis as it does not depend on $\OPT$ (which varies through iterations). 

In order to show the usefulness of the rounding procedure, we first present a simple $(1.7+\e)$-competitive algorithm using a migration factor of $O(1/\e)$. This algorithm maintains through the arrival of new jobs a locally optimal solution with respect to Jump for large jobs and a greedy assignment for small jobs on top of that. Although for general instances this can induce a very large migration factor as discussed before, for rounded instances we can have a very accurate control on the jumps needed to reach a locally optimal solution by exploiting the fact that there are constant many possible processing times for large jobs. 

In the second part of the paper we proceed with the analysis of our $(4/3+\e)$-competitive algorithm. Here we crucially make use of the properties obtained by the rounding procedure to create symmetries. After a new job arrival we re-run the LPT algorithm for the new instance. While assigning a job to a current least loaded machine, since there is a constant number of possible machine loads, there will usually be multiple least loaded machines to assign the job. All options lead to different (but symmetric) solutions in terms of job assignments, all having the same load vector and thus the same objective value. Broadly speaking, the algorithm will construct one of these symmetric schedules, trying to maintain as many machines with the same assignments as in the previous time step. The analysis of the algorithm will rely on monotonicity properties implied by LPT which, coupled with rounding, implies that for every job size the increase in the number of machines with different assignments (w.r.t the solution of the previous time step) is constant. This finally yields a migration factor that only grows polynomially in $1/\e$. 
Finally, we give a lower bound of $17/16$ for the best competitive ratio achievable by an algorithm with constant migration, improving the previous bound from Skutella \& Verschae~\cite{SV10}. 

\section{Preliminaries}\label{LPTsec}

Consider a set of $n$ jobs $\J$ and a set of $m$ machines $\M$. In our problem, a solution or schedule $\Sc:\J\rightarrow \M$ corresponds to an assignment of jobs to machines. The set of jobs assigned to a machine $i$ is then $\Sc^{-1}(i)\subseteq J$. The load of machine $i$ in $\Sc$ corresponds to $\ell_i(\Sc) = \sum_{j\in \Sc^{-1}(i)} p_j$. The minimum load is denoted by $\lmin(\Sc)=\min_{i\in \M} \ell_i(\Sc)$, and a machine $i$ is said to be \emph{least loaded} in $\Sc$ if $\ell_i(\Sc) = \lmin(\Sc)$.	

For an algorithm $\Alg$ and an instance $(\J,\M)$, we denote by $\Sc_{\Alg}(\J,\M)$ the schedule returned by $\Alg$ when run on $(\J,\M)$. Similarly, $\Sc_{\OPT}(\J,\M)$ denotes the optimal schedule, being $\OPT(\J,\M)$ its minimum load. When it is clear from the context, we will drop the dependency on $\J$ or $\M$.

\subsection{Algorithms with robust structure}

An important fact used in the design of the robust PTAS for makespan minimization from Sanders et al.~\cite{SSS09} is that small jobs can be assigned greedily almost without affecting the approximation guarantee. This is however not the case for machine covering; see, e.g.~\cite{SV10} or Section~\ref{sec:lowerbound}. One way to avoid this inconvenience is to develop algorithms that are oblivious to the arrival of small jobs, that is, algorithms where the assignment of big jobs is not affected when a new small job arrives.

\begin{definition}\label{EstRob} Let $h\in\mathbb{R}_+$. An algorithm $\Alg$ has \textbf{robust structure at level $h$} if, for any instance $(\J,\M)$ and $j^*\notin \J$ such that $p_{j^*}< h$, $\Sc_{\Alg}(\J,\M)$ and $\Sc_{\Alg}(\J\cup\{j^*\},\M)$ assign to the same machines all the jobs in $\J$ with processing time at least $h$. \end{definition}

This definition highlights also the usefulness of working with the $\LPT$ rule, since the addition of a new small job to the instance does not affect the assignment of larger jobs. Indeed, it is easy to see the following.

\begin{remark} 
	For any $h\in \mathbb{R}_+$, $\LPT$ has robust structure at level $h$.
\end{remark}

We proceed now to define \emph{relaxed} solutions where, roughly speaking, small jobs are added greedily on top of the assignment of big jobs.

\begin{definition}\label{k-rel} Let $\Alg$ be an $\alpha$-approximation algorithm for the machine covering problem, with $\alpha$ constant, $k_1, k_2\in \mathbb{R}_+$ constants, $1\le k_1\le k_2$ and $\e>0$. Given a machine covering instance $(\J,\M)$, a schedule $\Sc$ is a \textbf{$(k_1,k_2)$-relaxed version of $\Sc_{\Alg}$} if: 
	\begin{enumerate} 
		\item jobs with processing time at least $k_1\e\OPT$ are assigned exactly as in $\Sc_\Alg$, and
		\item for every machine $i\in \M$, if $\Sc$ assigns at least one job of size less than $k_1\e\OPT$ to $i$, then $\ell_i(\Sc) \le \lmin(\Sc) + k_2\e\OPT$.
	\end{enumerate}
\end{definition}

The following lemma shows that we can consider relaxed versions of known algorithms or solutions while almost not affecting the approximation factor. This will be helpful to control the migration of small jobs.

\begin{lemma}\label{RobGree} Let $\Alg$ be an $\alpha$-approximation, $\alpha\ge1$ constant, $k_1, k_2\in\mathbb{R}_+$ constants, $1\le k_1\le k_2$, $0<\e<\frac{1}{2k_2\alpha}$ and $(\J,\M)$ a machine covering instance. Every $(k_1,k_2)$-relaxed version of $\Sc_{\Alg}$ is an $(\alpha + O(\e))$-approximate solution. \end{lemma}

\begin{proof} Suppose by contradiction that there exists a $(k_1,k_2)$-relaxed version of $\Sc_{\Alg}$, say $\Sc$, which is not $(\alpha + 2k_2\alpha^2\e)$-approximate. This implies that $\lmin(\Sc) < \frac{1}{\alpha+2k_2\alpha^2\e}\OPT \le \left( \frac{1}{\alpha} - k_2\e\right) \OPT$.
	
	Let $\M_s$ the set of machines where $\Sc$ assigns at least one job of size less than $k_1\e\OPT$. Notice that $\M_s \neq \emptyset$ and actually the least loaded machine in $\Sc$ belongs to $\M_s$, because otherwise $\lmin(\Sc_{\Alg}) = \lmin(\Sc) < \left( \frac{1}{\alpha} - k_2\e\right) \OPT$, which contradicts that $\Sc_\Alg$ is $\alpha$-approximate. Since $\Sc$ and $\Sc_{\Alg}$ assign to the same machines jobs of size at least $k_1\e\OPT$, we have that the total processing time of jobs assigned by $\Sc$ to $\M_s$ is at most $\lvert \M_s \rvert(\lmin(\Sc) + k_2\e\OPT)$. Thus, \begin{equation*} \lmin(\Sc_{\Alg}) \le \displaystyle\min_{i\in\M_{s}}{\ell_i(\Sc_{\Alg})} \le \lmin(\Sc) + k_2\e\OPT < \frac{1}{\alpha}\OPT, \end{equation*} which contradicts that $\Sc_\Alg$ is $\alpha$-approximate. \end{proof}

The described results allow us to significantly simplify the analysis of the designed algorithms. For example, consider $\LPT$ and suppose that at the arrival of jobs with processing time at least some specific value $h=\Theta(\varepsilon \OPT)$ we can construct relaxed versions of solutions constructed by $\LPT$. Dealing with an arriving job of size smaller than $h$ becomes a simple task since assigning it to the current least loaded machine does not affect the assignment of big jobs, and we can prove that, for suitable constants $k_1, k_2$, a $(k_1,k_2)$-relaxed version of a solution constructed by $\LPT$ is maintained that way, almost preserving then its approximation ratio. It is important to remark that this approach is useful only if the algorithm has robust structure as, in general, the arrival of small jobs does not allow migration of big jobs and their structure may need to be changed because of these arrivals in order to maintain the approximation factor (see for example Section~\ref{sec:lowerbound}). 

\subsection{Rounding procedure}\label{sec:Rounding}

Another useful tool is rounding the processing times to simplify the instance and create symmetries while affecting the approximation factor only by a negligible value. Let us consider $0<\e<1$ such that $1/\e \in \mathbb{Z}$. We use the following rounding technique which is a slight modification of the one presented by Hochbaum and Shmoys in the context of makespan minimization on related machines~\cite{HS88}. For any job $j$, let $e_j\in \mathbb{Z}$ be such that $2^{e_j}\le p_j< 2^{e_j+1}$. We then round down $p_j$ to the previous number of the form $2^{e_j} + k\e 2^{e_j}$ for $k\in \mathbb{N}$, that is, we define $\tilde{p}_j := 2^{e_j} + \left\lfloor \frac{p_j-2^{e_j}}{\e2^{e_j}} \right\rfloor \e 2^{e_j}.$

Observe that $p_j\ge \tilde{p}_j \ge p_j - \e2^{e_j} \ge (1-\e)p_j$. Hence, an $\alpha$-approximation algorithm for a rounded instance has an approximation ratio of $\alpha/(1-\e)=\alpha+O(\varepsilon)$ for the original instance. From now on we work exclusively with the rounded processing times.   

Consider an upper bound $\UB$ on $\OPT$ such that $\OPT\le \UB \le 2\OPT$. This can be computed using any 2-approximation for the problem, in particular $\LPT$. Consider the index set
\begin{equation}
\label{eq:deftP1}	
\tilde{I}(\UB) := \left\lbrace i\in\mathbb{Z}:\e\UB \le 2^i < \UB\right \rbrace= \{\ell,\dots,u\}.
\end{equation} 

We classify jobs as \emph{small} if $\tilde{p}_j < 2^{\ell}$, \emph{big} if $\tp_j\in [2^{\ell},2^{u+1})$, and \emph{huge} otherwise. Notice that small jobs have size at most $2\e\UB$ and huge jobs have size at least $\UB$. As we will see, our main difficulty will be given by big jobs; small and huge jobs are easy to handle. Notice that in every solution $\Sc$ constructed using $\LPT$, if we ignore small jobs, huge jobs are assigned to a machine on their own and every machine $i\in \M$ without huge jobs has load at most $2\UB$.  This is because $i$ either has a big job alone, which has size at most $2\UB$, or it has load at most $\lmin(\Sc) + \tp_j \le 2\lmin(\Sc) \le 2 \UB$, where $j$ is the smallest job assigned to $i$. Let 
\begin{equation}
\label{eq:deftP2} \tP=\left\lbrace2^{i}+k\e2^{i}: i\in\{\ell,\ldots,u\}, k\in\{0,1,\ldots,(1/\e)-1\}\right\rbrace
\end{equation}
be the set of all (rounded) processing times that a big job may take. The next lemma highlights the main properties of our rounding procedure.

\begin{lemma}
	\label{lm:rounding}
	Consider the rounded job sizes $\tp_j$ for all $j$. Then it holds that,
	\begin{enumerate}
		\item $|\tP|\in O((1/\e)\log (1/\e))$, and
		\item for each big and huge job $j$ it holds that $\tp_j=h\cdot \e 2^{\ell}$ for some $h\in \mathbb{N}_0$. 
	\end{enumerate}
\end{lemma}

\begin{proof} From the definition of $\tP$, we have that $|\tP| = \frac{1}{\e}(u-\ell-1)-1$. Since $\ell \ge \log\left(\e\UB\right)$ and $u\le \log(\UB)$, then $(u-\ell-1) \le \log \left( \frac{1}{\e} \right)$. Altogether, $|\tP|\in O((1/\e)\log (1/\e))$. Also, if $j$ is a big or huge job, then $\tp_j = 2^{i}+k\e2^{i}$ for some $i\ge \ell$ and $k\in \{0,\dots, \frac{1}{\e}-1\}$. We conclude by noticing that $\tp_j = \left( \frac{1}{\e} + k \right)2^{i-\ell} \cdot \e2^{\ell} = h \cdot \e 2^{\ell}$ for $h = \left( \frac{1}{\e} + k \right)2^{i-\ell} \in \mathbb{N}_0$. \end{proof}

Unlike other standard rounding techniques (e.g. \cite{SV10, JKV16}), the rounded sizes do not depend on $\OPT$ (or $\UB$). This avoids possible migrations provoked by new rounded values, greatly simplifying our techniques. 

\section{A simple $(1.7+\e)$-competitive algorithm with $O(1/\e)$ migration.}\label{sec:JOpt}

In this section we will adapt a local search algorithm for Machine Covering to the online context with migration, using the properties of instances rounded as described in Section~\ref{sec:Rounding} to bound the migration factor.

In the context of online load balancing with migration, it is a good strategy to look for local search algorithms with good approximation guarantees and efficient running times. The main reason is that the migrated load corresponds to the sum of the migrated jobs in each local move, and for simplified instances (rounded, for example) the number of local moves until a locally optimal solution is found is usually a constant. That is the case for two natural neighborhoods used in local search algorithms for load balancing problems: \emph{Jump} and \emph{Swap}. 
Two solutions $\Sc, \Sc'$ are \emph{jump-neighbors} if they assign the jobs to the same machines (up to relabeling of machines or jobs of equal size) except for at most one job, and \emph{swap-neighbors} if they assign the jobs to the same machines (up to relabeling of machines or jobs of equal size) except for at most two jobs and, if they differ in exactly two jobs $j_1,j_2$ then they are in swapped machines, i.e., $\Sc(j_1)=\Sc'(j_2)$ and $\Sc(j_2)=\Sc'(j_1)$. The \emph{weight} of a solution is defined through a two-dimensional vector having the minimum load of the schedule as first coordinate and the number of non-least loaded machines as second one. We compare the weight of two solutions lexicographically\footnote{Just using the minimum load does not lead to good approximation ratios: think for example of $m>2$ machines and $m$ jobs of size $1$; it is swap-optimal to assign all of them to the same machine.}. In other words, a solution is jump-optimal (respectively swap-optimal) if the migration of a single job (resp.~the migration of a job or the swapping of two jobs) does not increase the minimum load and, if it maintains the minimum load, then it does not reduce the number of least loaded machines. The following lemma characterizes jump-optimal solutions for machine covering. \jvcom{No se si el ``up to symmetry'' se entiende todavía... que simetrías? Maquinas? Trabajos del mismo tamaño? Aquí hay que explicar.}\wgcom{Lo especifiqué mejor}

\begin{lemma}\label{CarOL} Given $(\J,\M)$ a machine covering instance, a schedule $\Sc$ is jump-optimal if and only if for any machine $i\in\M$ and any job $j\in\Sc^{-1}(i)$, we have that $\ell_i(\Sc)-p_j \le \lmin(\Sc)$. \end{lemma}

\begin{proof} If $\Sc$ is jump-optimal and there is a job not satisfying the inequality, then moving it to a least loaded machine either increases the minimum load of the schedule or reduces the number of least loaded machines, which is a contradiction.
	
On the other hand, if $\Sc$ is not jump-optimal then there is a job $j$ whose migration improves the weight of the solution or, if not, a job whose migration to a least loaded machine decreases the number of least loaded machines. Consider first the case in which moving $j$ from a machine $i$ increases the minimum load. This means that the new load of machine $i$, $\ell_i(\Sc)-p_j$, is at least the new minimum load, which is strictly larger than $\lmin(\Sc)$, proving the needed inequality. Consider now the case in which moving $j$ from a machine $i$ to machine $i'$ maintains the minimum load while reducing the number of least loaded machines. Then $i'$ must have been a non-unique least loaded machine in $\Sc$. Furthermore, machine $i$ cannot become a minimum loaded machine in the new schedule (otherwise the number of least loaded machines would not change). This means that $\ell_i(\Sc)-p_j$ is strictly larger than $\lmin(\Sc)$. \end{proof}

Chen et al.~\cite{CEKvS13} proved tight bounds for the approximability of jump-optimal solutions. Their result is stated in a game theoretical framework, where jump-optimal solutions are equivalent to pure Nash equilibria for the Machine Covering game (see for example~\cite{V07}). In this game, each job is a selfish agent trying to minimize the load of its own machine and the minimum load is the welfare function to be maximized. Through a small modification these bounds can be generalized to swap-optimal solutions as well (notice that a swap-optimal solution is jump-optimal by definition). We summarize the result in the following theorem which will be useful for our purposes (refer to Appendix~\ref{sec:chenetal} for details).

\begin{theorem}[from~\cite{CEKvS13}]\label{thm:1.7apx} Any locally optimal solution with respect to Jump (resp. Swap) for Machine Covering is $1.7$-approximate. Moreover, there are instances showing that the approximation ratio of jump-(resp. swap-)optimality is at least $1.7$. \end{theorem}

\subsection{Online jump-optimality.}

Using the rounding procedure developed in Section \ref{sec:Rounding}, jump-optimality can be adapted to the online context using migration factor $O\left(\frac{1}{\e}\right)$. In order to do this, we need an auxiliary algorithm called \emph{Push} (Algorithm~\ref{push}) to assign a job $j$ to a given machine. This procedure inserts a given job to a given machine, and then iteratively removes the jobs that break jump-optimality according to Lemma~\ref{CarOL}, storing them in a special set $Q$ which is part of its output. This algorithm is the base of the Push neighborhood analyzed by Schuurman and Vredeveld \cite{SV07}.

Our algorithm, described in detail in Algorithm~\ref{OLSonline}, is called every time a new job $j^*$ arrives to the system, and receives as input the current solution $\Sc$ for $(\J,\M)$, initialized as empty if $\J=\emptyset$. It will output a $(k,2k)$-relaxed \jvcom{Que significa $4$-relaxed? La versión actual necesita dos constantes en la definión...}\wgcom{Corregido} version of a jump-optimal solution for some $k\le 4$. We use the concept of a \emph{list-scheduling} algorithm, that refers to assigning jobs iteratively (in any order) to some machine of minimum load. Given a schedule $\Sc$, $\Sc_B$ denotes the restriction of schedule $\Sc$ to big jobs.

\begin{algorithm}[h!t]
	\caption{Push}
	\label{push}
	\begin{algorithmic}[1]
		\Require
		\text{ Schedule $\Sc$ for $(\J,\M)$, $i\in \M$, $j\notin \J$}
		\Ensure \text{ $Q\subseteq \J$, schedule $\Sc'$ for $((\J\cup\{j\})\setminus Q,\M)$}
		\State $Q \gets \emptyset$.
		\State $\Sc' \gets \Sc$.
		\State assign $j$ to machine $i$ in $\Sc'$.
		\For{$k \in \Sc^{-1}(i)$}
		\If{$\ell_i(\Sc')-\tp_k>\lmin(\Sc')$}
		\State take out $k$ from $i$ in $\Sc'$.
		\State $Q \gets Q \cup \{k\}$.
		\EndIf
		\EndFor
		\State \Return $Q$, $\Sc'$.
	\end{algorithmic}
\end{algorithm}

The general idea of Algorithm~\ref{OLSonline} is to first round the instance, and assign the incoming job to a least loaded machine using Algorithm~\ref{push}. Jobs removed by Algorithm~\ref{push} need to be reassigned, which we do by iteratively applying Algorithm~\ref{push} on each one of them which is big until only small jobs are left to be assigned. At each iteration jump-optimality is preserved in a relaxed way, and as a last step all remaining small jobs are reassigned using list-scheduling.
Notice that, since Algorithm~\ref{push} only removes jobs of size strictly smaller than the inserted job, each job is migrated at most once.

\begin{algorithm}[h!]
	\caption{Online jump-optimality}
	\label{OLSonline}
	\algorithmicrequire{ \parbox[t]{12.75cm}{ 
			Instances $(\J,\M)$ and $(\J',\M)$ such that $\J'=\J\cup\{j^*\}$; a schedule $\Sc(\J,\M)$.}}
	\begin{algorithmic}[1]
		
		\State{run LPT on input $\J'$ and let $\tau$ be the minimum load. Set $\UB\gets 2\tau$. Define $\tP, \ell$, and $u$ based on this upper bound $\UB$ using \eqref{eq:deftP1} and \eqref{eq:deftP2}.}
		\State{set $\Sc' \gets \Sc$}
		\If{$\tp_{j^*} < 2^\ell$}. 
		\State{assign $j^*$ to a least loaded machine in $\Sc'$.}
		\Else
		\State{set $Q_B \gets \{j^*\}$.}\Comment{Set with unassigned big jobs.}
		\State{set $Q_s \gets \emptyset$.} \Comment{Set with unassigned small jobs.}
		\While{$Q_B \neq \emptyset$}
		\State{let $j$  be the largest job in $Q_B$. Set $Q_B \gets Q_B \setminus\{j\}$.}
		\State{in $\Sc_B'$, use Push (Algorithm \ref{push}) to assign $j$ to a least loaded machine $m^*$, obtaining its output set $Q$. Update $\Sc_B'$ to be the output solution of this procedure.}
		\State reassign jobs in $\Sc'$ such that the assignment of (big) jobs in $\Sc'$ and $\Sc_B'$ coincides.
		\While{$m^*$ contains a small job w.r.t. $\UB$ and $\ell_{m^*}(\Sc')>\ell_{\min}(\Sc')+2^\ell$}
		\State remove the smallest job in $\Sc'^{-1}(m^*)$ and add it to $Q_s$.
		\EndWhile
		\State $Q_B \gets Q_B \cup Q$.
		\EndWhile
		\State assign the jobs in $Q_s$ to $\Sc'$ using list-scheduling.
		\EndIf
		\State \Return $\Sc'$.
	\end{algorithmic}
\end{algorithm}
\jvcom{Falta definir $\Sc_B$ en el algoritmo.}\wgcom{corregido}

\begin{lemma}\label{lm:1.7-competitive}
	For any $h\in \mathbb{R}^+$, Algorithm~\ref{OLSonline} has robust structure at level $h$. Furthermore, Algorithm~\ref{OLSonline} is $(1.7+O(\e))$-competitive and has polynomial running time.
\end{lemma}

\begin{proof}
	First of all, Algorithm~\ref{OLSonline} has, for any $h\in\mathbb{R}^+$, robust structure to level $h$ because each time that Push is called, it moves a total load of jobs smaller than the processing time of the inserted job (otherwise, the machine would not be a least loaded machine), and if the job is small, then nothing is migrated. This also directly implies that the running time of the algorithm is polynomial because every job is migrated at most once, so the while loop is executed only a polynomial number of times.
	
	In order to show that the competitive ratio is $(1.7+O(\e))$, we just need to show that the schedule constructed by Algorithm \ref{OLSonline} is a $(k_1,k_2)$-relaxed version of a jump-optimal solution for some constants $k_1, k_2$. Having that, the result follows from Theorem~\ref{thm:1.7apx} and Lemma~\ref{RobGree}.
	
	Let $k=\frac{2^\ell}{\e\OPT'}$. We will prove that the constructed schedule is a $(k,2k)$-relaxed version of a jump-optimal schedule by induction on $|\J|$ (notice that $k$ depends on $\OPT'$ and $\ell$, hence on the instance $(\J,\M)$). The base case when $\J=\emptyset$ is trivial. Let $\ell^{(1)}$ be the lower bound computed for $\OPT$ and $k^{(1)}=\frac{2^{\ell^{(1)}}}{\e\OPT}$, and let us assume that $\Sc$ is a $(k^{(1)},2k^{(1)})$-relaxed version of some jump-optimal solution $\Sc^*$ for $(\J,\M)$ (recall that $\OPT\le \OPT'$ and $\ell^{(1)}\le \ell$). This means that $\Sc$ and $\Sc^*$ assign to the same machines jobs of size at least $2^{\ell^{(1)}}$, and machines in $\Sc$ containing at least one job of size smaller than $2^{\ell^{(1)}}$ have load at most $\ell_{\min}(\Sc)+2\cdot2^{\ell^{(1)}}$. Our goal is to prove that the output $\Sc'$ of Algorithm~\ref{OLSonline} when run on $\Sc(\J,\M)$ plus an arriving job $j^*$ is a $(k,2k)$-relaxed version of some jump-optimal solution $\Sc^{**}$ for $(\J\cup\{j^*\},\M)$.
	
	Notice first that for this $k$ we have that big jobs have processing time at least $k\e\OPT'$. If $\tp_{j^*}<2^\ell$, it is easy to see that the conditions are fulfilled since it is assigned to the least loaded machine. Assume from now on that $j^*$ is big. Suppose that we run Algorithm~\ref{OLSonline} on $\Sc^*(\J,\M)$ and arriving job $j^*$, getting a solution $\Sc^*_{aux}$. First of all it is not difficult to see that the minimum load does not decrease when applying Algorithm~\ref{OLSonline}. Thanks to the jump-optimality of $\Sc^*$ we have that, for every machine $i$ where no job was assigned using Push and any job $j$ assigned to $i$, $\ell_i(\Sc^*_{aux})-p_j<\lmin(\Sc^*_{aux})$, and hence the jobs breaking jump-optimality in $\Sc^*_{aux}$ can only belong to the remaining machines. In these machines we either have only big jobs or they have load at most $\lmin(\Sc^*_{aux})+2^\ell$, implying that the jobs breaking jump-optimality are small thanks to Lemma~\ref{CarOL}. If we take out from the solution such jobs and reassign them using Push until no job is left to be assigned (i.e. reassigning also the jobs which are pushed out) we get a jump-optimal solution $\Sc^{**}$. Since this procedure moves only small jobs (as pushed jobs are always smaller than the assigned job), the assignment of big jobs in $\Sc'$ and $\Sc^{**}$ is the same, proving the first part of being a $(k,2k)$-relaxed version of some jump-optimal solution.
	
	We will now prove that if a machine has at least one job of size at most $2^\ell$ then its load is at most $\lmin(\Sc')+2\cdot 2^\ell$. To this end we will consider three cases:
	\begin{itemize} \item If $i$ is a machine where no job was assigned using Push and it has a job of size smaller than $2^{\ell^{(1)}}$, since $\Sc$ is a $(k^{(1)},2k^{(1)})$-relaxed version of some jump-optimal solution, the load of $i$ is at most $\lmin(\Sc)+2\cdot 2^{\ell^{(1)}} \le \lmin(\Sc')+2\cdot 2^\ell$.
		\item If $i$ is a machine where no job was assigned using Push, it has only jobs of size at least $2^{\ell^{(1)}}$ and has at least one job of size smaller than $2^\ell$ (implying that $\ell^{(1)}<\ell$), since $\Sc$ is a $(k^{(1)},2k^{(1)})$-relaxed version of some jump-optimal solution $\Sc^*$, the load of $i$ is at most $\lmin(\Sc)+2\cdot2^{\ell^{(1)}} \le \lmin(\Sc^*)+2^{\ell}$. From the proof of Lemma~\ref{RobGree}, we have that $\lmin(\Sc^*)\le \lmin(\Sc')+2\cdot 2^{\ell^{(1)}}$. Putting everything together, we have that \begin{eqnarray*} \ell_i(\Sc') & \le & \lmin(\Sc^*)+2^\ell \\ & \le & \lmin(\Sc')+2\cdot 2^{\ell^{(1)}} + 2^\ell \\ & \le & \lmin(\Sc')+2\cdot 2^\ell, \end{eqnarray*} where the last inequality comes from the fact that $\ell^{(1)}<\ell$.
		\item If $i$ is a machine where some job was assigned using Push and it has at least one job of size smaller than $2^\ell$, the algorithm enforces its load to be at most $\lmin(\Sc')+2^\ell$.
	\end{itemize}
	
	This proves that $\Sc'$ is a $(k,2k)$-relaxed version of some jump-optimal solution, and we conclude the proof by noticing that $1\le k = 2^{\ell}/(\e\OPT')\le 2\e\UB/(\e \OPT') \le 4$.
\end{proof}

Now we will bound the migration factor, and also construct an instance showing that the analysis of the migration factor is essentially tight.

\begin{lemma}\label{OLS_migration} Algorithm~\ref{OLSonline} has migration factor $O\left(1/\e\right)$.\end{lemma}

\begin{proof} To analyze the migration factor, we define the \emph{migration tree} of the algorithm as a node-weighted tree $G=(V,E)$, where $V$ is the set of migrated jobs together with the incoming job $j^*\notin\J$, and the weight of each $v\in V$ is the processing time of the corresponding job $\tp_v$. The tree is constructed by first adding $j^*$ as root. For each node (job) $v$ in the tree, its children are defined as all the jobs migrated at the insertion of $v$. It is easy to see that this process does not create any loops as each job is migrated at most once. By definition, the leaves of the tree are the jobs not inducing migration, and thus any small job in the tree is a leaf. In the context of local search, the number of nodes in the tree corresponds to the number of iterations of the specific local search procedure. 
	
Let $w_i$ be the total processing time of nodes corresponding to big jobs in level $i$ of the migration tree. Assume that $\tp_{j^*} = q_{\kappa} =2^g + h\e2^g$ for some $\kappa \in \{1,\dots,|\tP|\}$, $g\in\{\ell,\dots,u\}$ and $h\in\{0,\dots,\frac{1}{\e}-1\}$. Every time a job $j$ is inserted using Push, the total load of jobs in the output $Q$ of the algorithm is strictly less than $\tp_{j}$, which means that $w_i$ is strictly decreasing, and also that at each level $i$ of the tree there are at most $\frac{w_i}{2^\ell}$ nodes corresponding to big jobs. Since the second condition of being a $(k_1,k_2)$-relaxed version (Definition~\ref{k-rel}) of a jump-optimal solution is maintained through the iterations, the small jobs that need to be migrated because of insertion of a big job $j$ have total load at most $\tp_{j}+2^\ell$. This implies that the total load of small jobs at each level $i\ge 1$ of the tree is at most $w_{i-1} + \frac{w_{i-1}}{2^\ell}\cdot 2^\ell = 2w_{i-1}$, and hence the total processing time of nodes corresponding to small jobs is at most twice the total processing time of nodes corresponding to big jobs. Because of that, from now on we will assume that the migration tree contains only nodes corresponding to big jobs. 
	
	We categorize each level $i\ge 1$ of the migration tree according to the following two cases: if there is a node in level $i-1$ having at least two children, we say that level $i$ falls in case $1$, and it falls in case $2$ otherwise. We first show that there are at most $\frac{\tp_{j^*}}{2^\ell}\le1/\e$ levels of the tree falling in case $1$. Because of the way the migration tree is constructed, it is not difficult to see that the total weight of the leaves in the tree is at most $\tp_{j^*}$ (this property is maintained inductively through the executions of Algorithm Push). Because of this, since each big job has processing time at least $2^\ell$, every migration tree has at most $\tp_{j^*}/2^\ell$ leaves, which is also an upper bound for the number of nodes that have more than one children in the tree (each one of them induces at least one extra leaf), and hence for the number of levels falling in case $1$.
	
	There can be more than $1/\e$ levels falling in case $2$ along the tree, but we will show that in that case $w_{i}$ quickly decreases based on the following claim.
	
	\noindent\textbf{Claim:} Let $q_{i_1}, \dots, q_{i_k} \in \tP$ such that $\displaystyle\sum_{j=1}^k{q_{i_j}} \in (q_{s+1}, q_s]$ for some $s\in \{1,\dots,|\tP|\}$. Then $\displaystyle\sum_{j=1}^k{q_{i_j+1}} \le \displaystyle\sum_{j=1}^k{q_{i_j}} - \frac{\e}{4}q_{s+1}$, where we assume that $q_{|\tP|+1} = 0$.
	
	Notice that the claim implies that for a level $i$ falling in case $2$, if $w_{i-1} \in (q_{s+1}, q_s]$ for some $s\in\{1,\dots,|\tP|\}$, then $w_i \le w_{i-1} - \frac{\e}{4}q_{s+1}$. To compute the total processing time of the nodes in the migration tree, we will bound the total weight of the levels corresponding to each case separately. Since there are at most $1/\e$ levels falling in case~$1$, each one of them having total weight at most $\tp_{j^*}$, we can bound the total weight of those levels by $\frac{1}{\e}\tp_{j^*}$. Let us now relabel the levels of the tree where the second case occurs by just $\{1,2,\dots,L_2\}$ (i.e. we ignore the levels falling in case $1$). Thanks to the claim, for every $i\in \{1,2,\dots,L_2\}$, if $w_{i-1} \in (q_{s+1},q_s]$ for some $s \in \{1,\dots,|\tP|\}$, then $\displaystyle\sum_{j=i-1}^{i+2}{w_j}\le 4q_s$ and $w_{i+3}\le q_{s+1}$ (because $q_s - \e q_{s+1} \le q_{s+1}$), and we can restart the process for $i+3$ with the correct value $q_{s'}\le q_{s+1}$. If we use this argument starting with $w_0 \in (q_{\kappa+1},q_{\kappa}]$, we can conclude that $\displaystyle\sum_{i=0}^{L_2}{w_i} \le 4\displaystyle\sum_{i=\kappa}^{|\tP|}{q_i}$, which, recalling that $\tp_{j^*} = 2^g + h\e2^g$, is at most
	\begin{align*} &
	4\sum_{i=\ell}^{g}{ \sum_{b=0}^{\frac{1}{\e}-1}{(2^i + b\e2^i)}} = 4\sum_{i=\ell}^{g}{2^{i-1}\left(\dfrac{3}{\e}-1\right) } \le 4\tp_{j^*}\left(\dfrac{3}{\e}-1\right).
	\end{align*}
	These two bounds, together with the fact that the total load of small migrated jobs is at most twice this value, implies that the migration factor is at most $O\left(\frac{1}{\e}\right)$.
	
	To prove the claim, notice that $\displaystyle\sum_{j=1}^k{q_{i_j+1}} \le \displaystyle\sum_{j=1}^k{q_{i_j}} - \displaystyle\sum_{j=1}^k{\e 2^{\lfloor\log(q_{i_j})\rfloor-1}} \le \displaystyle\sum_{j=1}^k{q_{i_j}} - \frac{\e}{2}\displaystyle\sum_{j=1}^k{2^{\lfloor\log(q_{i_j})\rfloor}}$. Also, since $\displaystyle\sum_{j=1}^k{2^{\lceil\log(q_{i_j})\rceil}}\ge \displaystyle\sum_{j=1}^k{q_{i_j}}>q_{s+1}$, we have that $\displaystyle\sum_{j=1}^k{2^{\lfloor\log(q_{i_j})\rfloor}}>\frac{q_{s+1}}{2}$. This concludes the proof of the claim. \end{proof}

\begin{lemma}\label{lm:lbOLS_migration}There are instances for which Algorithm \ref{OLSonline} uses a migration factor of at least $\Omega\left( \frac{1}{\e}\right)$. 	
\end{lemma}

\begin{proof} Consider an instance with $\OPT= 2^{u+1}$ and $\e\OPT= 2^{\ell}$ for some integers $\ell, u$, and assume for simplicity that $\UB= \OPT$. This way, $\tilde{I}(\UB) = \{\ell,\dots,u\}$. The instance, consisting of $m\in O\left( \frac{1}{\e}\log\frac{1}{\e}\right)$ machines, is constructed in the following way: Consider the possible processing times sorted non-increasingly $t_1,\dots,t_h$. For each $i$ such that $t_i < 2^u$, the schedule has a machine with a job of size $t_i$ assigned, and it is completed with jobs until having load $2^{u+1}$: if $t_i = 2^k + j\e2^k$, this can be done adding a job of size $2^k + \left(\frac{1}{\e}-j\right)\e2^k$, a job of size $2^k$ and for each $k' = k+2,\dots,u$, a job of size $2^{k'}$ (if $i=u-1$, the machine will not have any of these last jobs). By doing so, the load of the machine is \begin{equation*} 2^i + j\e2^i + 2^i + \left(\frac{1}{\e} - k\right)\e 2^i + 2^{i} + \displaystyle\sum_{i'=i+2}^{u}{2^{i'}} = 2^{i+2} + 2^{u+1}-2^{i+2} = 2^{u+1}. \end{equation*}
	
	Now, if a job of size $2^u$ arrives to the system, it can be inserted using Push in the machine with the largest job of size less than $2^u$ (i.e., with processing time $2^{u-1} + \left(\frac{1}{\e}-1\right)\e2^{u-1}$), taking out such job because it breaks jump-optimality. If Algorithm \ref{OLSonline} takes the decision in the same way iteratively, then at least one job of each possible size $t_i<2^u$ is migrated, being then the total migrated load at least \begin{equation*} \displaystyle\sum_{i=\ell}^{u-1}{\displaystyle\sum_{j=1}^{\frac{1}{\e}-1}{2^i + j\e2^i}} = \left(\dfrac{1}{\e}-1\right)(2^u-2^{\ell+1}) + \dfrac{1}{2}\left(\dfrac{1}{\e}-1\right)2^u \in \Omega\left(\dfrac{1}{\e}2^u\right), \end{equation*} and hence the migration factor needed for this instance is $\Omega\left( \frac{1}{\e} \right)$. \end{proof}

By putting together Lemmas~\ref{lm:1.7-competitive}, \ref{OLS_migration} and \ref{lm:lbOLS_migration} we can conclude the following result.

\begin{theorem}\label{hinfOnl} Given $\e>0$, Algorithm \ref{OLSonline} is a polynomial time $(1.7+\e)$-competitive algorithm and uses migration factor $O\left(1/\e\right)$. Moreover, there are instances for which this factor is $\Omega\left(1/\e\right)$. \end{theorem}

\section{LPT online with migration $\tilde{O}(1/\e^3)$.}\label{sec:4/3online}

In this section we present our main contribution which is an approximate online adaptation of $\LPT$ using $\text{poly}(1/\e)$ migration factor. In order to analyze it, we will first show some structural properties of the solutions constructed by $\LPT$ and how they behave when the instance is perturbed by a new job. 

Algorithm~\ref{OLSonline} presented in Section~\ref{sec:JOpt} already gives some of the features and properties that our online version of $\LPT$ fulfills. However, now in the analysis we will crucially exploit the symmetry of instances rounded according to the procedure described in Section~\ref{sec:Rounding}, in particular the fact that the load of each machine is a multiple of some fixed value. Since $\LPT$ takes decisions based solely on the machine loads, having a bounded number of values for them allows us to accurately control the set of machines where the assignment of big jobs can be kept unchanged after the arrival of a big job while maintaining the structure of the solution.\wgcom{Revisar este párrafo, aquí hay que enfatizar la importancia de la sección 4}\jvcom{Me parece ok!} Unless stated otherwise, for the rest of this section machine loads are considered with respect to the rounded processing times $\tp_j$.

\paragraph*{\textbf{Load Monotonicity.}} Here we describe in more detail the useful structural properties of solutions constructed using $\LPT$.

\begin{definition}\label{LoaPro} Given a schedule $\Sc$, its \textbf{load profile}, denoted by $\load(\Sc)$, is an $\mathbb{R}_{\ge0}^m$-vector $(t_1,\ldots,t_m)$ containing the load of each machine sorted so that $t_1\le t_2 \le \ldots\le t_m$.\end{definition}

The following lemma shows that after the arrival of a job, the load profile of solutions constructed using $\LPT$ can only increase. This property only holds if the vector of loads is sorted, as it can be seen in Figure~\ref{fig:LPT_no_red}. This monotonicity property is essential for our analysis. To show the mentioned property the following rather technical lemma will help.

\begin{lemma}\label{loadprof_generalized} Let $x,y\in \mathbb{R}_+^n$, $x=(x_1,x_2,\dots,x_n)$, $y=(y_1,y_2,\dots,y_n)$ such that $x_1\le x_2 \le \dots\le x_n$, $y_1\le y_2\le \dots\le y_n$ and $x\le y$ coordinate-wise, and $\alpha, \beta\in \mathbb{R}$ such that $\alpha\le \beta$. If we consider the new vectors defined by replacing $x_i$ by $x_i + \alpha$ in $x$ and $y_i$ by $y_i + \beta$ in $y$ for some $i\in\{1,2,\dots,n\}$, and then we sort the coordinates non-decreasingly of the new vectors, obtaining $x'$ and $y'$, then $x' \le y'$ coordinate-wise. \end{lemma}

\begin{proof} Let $\bar{i}$ be the coordinate such that $x'_{\bar{i}}=x_i+\alpha$ and $\bar{j}$ such that $y'_{\bar{j}}=y_i+\beta$. For each coordinate $k<\min\{\bar{i},\bar{j}\}$ or $k>\max\{\bar{i},\bar{j}\}$ we have that $x'_k=x_k$ and $y'_k=y_k$, thus satisfying the desired inequality by hypothesis. In the remaining we have two cases:
	
	\begin{itemize} \item Assume that $\bar{i}<\bar{j}$, and let $k\in\{\bar{i},\bar{i}+1,\ldots,\bar{j}-1\}$. Then we have that
		\[
		x'_k \le x'_{k+1} = x_{k+1} \le y_{k+1} = y'_k 
		\]
		where the first inequality holds due to the monotonicity of the vectors and the second one because of the hypothesis. Similarly, for $k=\bar{j}$ we have that
		\[
		x'_k = x_k \le y_k = y'_{k-1} \le y'_k, 
		\]
		where the first inequality follows from the hypothesis. 
		
		\item Assume that $\bar{i}\ge \bar{j}$ and let $k\in\{\bar{j},\bar{j}+1,\ldots,\bar{i}\}$. Then,
		\begin{equation*}
		x'_k\le x'_{\bar{i}} \le y'_{\bar{j}} \le y'_k,
		\end{equation*}
		where, the first and third inequalities follows from the monotonicity of the vectors, and the second one from the fact that $x_i+\alpha \le y_i + \beta$. \qedhere \end{itemize} 
\end{proof}

\begin{lemma}\label{load} Let $(\J,\M)$ be a machine covering instance and $j^*\notin\J$ a job. Then, it holds that \( \load(\Sc_{\LPT}(\J,\M)) \le \load(\Sc_{\LPT}(\J',\M)),\) where the inequality is considered coordinate-wise and $\J' = \J\cup\{j^*\}$. \end{lemma}

\begin{proof} Let us first relabel the jobs in $\J$ so that $\tp_1\ge \tp_2\ge \ldots \ge \tp_n$. To simplify the argument we assume that both runs of $\LPT$ assign jobs in the order given by the labeling above $1,2,\ldots,n$, where in the run for $\J'$ the new job $j^*$ is inserted to the list in any position consistent with $\LPT$. This is without loss of generality since different tie breaking do not affect the load profiles of the solutions.

Consider the set of instances $(\J|_k,\M)$ for $k=r,\dots,n$, where $\J|_k \subseteq \J$ is the set of the $k$ largest jobs in $\J$, and $r$ is the maximal index such that $\tp_{r}\le \tp_{j^*}$. Similarly, let $\J'|_k = \J|_k\cup \{j^*\}$ for any $k\in \{r,\ldots,n\}$. We will show by induction that the lemma is true for each pair $(\J|_k,\M)$ and $(\J'|_k,\M)$. The base case $k=r$ follows easily from Lemma~\ref{loadprof_generalized} since $\Sc_{\LPT}(\J|_k,\M)$ and $\Sc_{\LPT}(\J'|_k\setminus\{j^*\},\M)$ assign to the same machines all jobs $\{1,\ldots,r\}$, and adding $j^*$ to the least loaded machine in $\Sc_{\LPT}(\J'|_k\setminus\{j^*\},\M)$ (and a job of size $0$ to the least loaded machine in $\Sc_{\LPT}(\J|_k,\M)$) is the same as adding $\tp_{j^*}$ to the first coordinate of $\load(\Sc_{\LPT}(\J'|_k\setminus\{j^*\},\M))$, and then the inequality holds.

Suppose now that $\load(\Sc_{\LPT}(\J|_k,\M)) \le \load(\Sc_{\LPT}(\J'|_k,\M))$. Showing that the inequality is true for $k+1$ is equivalent to show that when assigning job $k+1$ to a least loaded machine in $\Sc_{\LPT}(\J|_k,\M)$ and in $\Sc_{\LPT}(\J'|_k,\M)$, the resulting load profiles satisfy the inequality, which is precisely the statement of Lemma~\ref{loadprof_generalized} adding $\tp_{k+1}$ to the first coordinate of $\load(\Sc_{\LPT}(\J|_k,\M))$ and also to the first coordinate of $\load(\Sc_{\LPT}(\J'|_k,\M))$. \end{proof}

This lemma together with our rounding procedure allow us to show that the difference (in terms of the Hamming distance) of the load profiles of two consecutive solutions consisting purely of big jobs, is bounded by a small constant. This property will be important to obtain a $\text{poly}(1/\e)$ migration factor and here we crucially exploit the fact that the load of the machines is always multiple of a fixed value.

\begin{lemma}\label{PocaSub} Consider two instances $(\J,\M)$ and $(\J',\M)$ with $\J'=\J\cup\{j^*\}$, where $\J'$ contains only big or huge jobs w.r.t $\UB$. Then the vectors $\load(\Sc_{\LPT}(\J,\M))$ and $\load(\Sc_{\LPT}(\J',\M))$ differ in at most $\frac{\tp_{j^*}}{\e 2^{\ell}} \in O(1/\e^2)$ many coordinates.\end{lemma}

\begin{proof} We have that $\load(\Sc_{\LPT}(\J,\M))=(t_1,\ldots,t_m)\le (t_1',\ldots,t'_m)=\load(\Sc_{\LPT}(\J',\M))$ thanks to Lemma~\ref{load}. Also, if $t_i< t_i'$ for some $i$, then $t_i'\ge t_i + \e 2^{\ell}$ since all values $t_{j},t_{j'}$ are integer multiples of $\e 2^{\ell}$ because of Lemma~\ref{lm:rounding}. Since $||\load(\Sc_{\LPT}(\J',\M))-\load(\Sc_{\LPT}(\J,\M))||_1= \tp_{j^*}$, we obtain that the number of coordinates in which the load profiles differ is at most $\frac{\tp_{j^*}}{\e 2^{\ell}}$.  Finally, recalling that $j^*$ is big, then $\tp_{j^*}\le 2^u \le \UB \le 2^{\ell}/\e$, and we can bound the number of different coordinates by $\frac{\tp_{j^*}}{\e 2^{\ell}} \le 1/\e^2$. 
\end{proof}

\paragraph*{\textbf{Description of Online LPT.}}

Consider two instances $(\J,\M)$ and $(\J',\M)$ such that $\J'=\J\cup\{j^*\}$, and let $\OPT$ and $\OPT'$ be their optimal values respectively. In what follows, for a given list-scheduling algorithm, we will refer to a tie-breaking rule as a rule that decides a particular machine for assigning a job when faced with multiple least loaded machines. We say that an assignment is an LPT-solution if there is some tie-breaking rule such that LPT yields such assignment. We will compute an upper bound $\UB$ on $\OPT'$ by computing an LPT-solution and duplicating the value of its minimum load. For this upper bound, we compute its respective set $\tP$ with \eqref{eq:deftP1} and \eqref{eq:deftP2}. In the algorithm, we will label elements in $\tP=\{q_1,\ldots,q_{|\tP|}\}$ such that $q_1 > q_2 > \cdots> q_{|\tP|}$. 
Let $\J_h\subseteq \J$ (respectively $\J_h'\subseteq \J'$) be the set of jobs of size $q_h$ in $\J$ (respectively $\J'$), for $q_h\in \tP$. Similarly, we define $\J_0$ (resp. $\J'_0$) to be the set of jobs in $\J$ (resp. $\J'$) of sizes larger than $q_1$, that is, all huge jobs in $\J$ (resp. $\J'$). 
Also, let $\Sc_{h}$ (resp. $\Sc_h'$) be the solution $\Sc$ (resp. $\Sc'$) restricted to jobs of size $q_h$ or larger. Finally, $\Sc_0$ and $\Sc_0'$ are the respective solutions restricted to jobs in $\J_0$.

In what follows, $x_+$ denotes the positive part of $x\in\mathbb{R}$, i.e., $x_+=\max\{x,0\}$. To understand the algorithm, it is useful to have the following observation in mind.

\begin{observation}\label{obs:LScharac}
	Consider a solution $\Sc$ for jobs in $\J$ and let $\mathcal{K}$ be a set of jobs with $\J\cap \mathcal{K}=\emptyset$ and all jobs in $\mathcal{K}$ have the same size $p$. Consider a solution $\Sc_{LS}$ constructed by adding the jobs from $\mathcal{K}$ in $\Sc$ using list-scheduling, and let $\lambda=\ell_{\min}(\Sc_{LS})$. Notice that $\lambda$ is independent of the tie-breaking rule used in list-scheduling. Consider any solution $\Sc'$ that is constructed starting from $\Sc$ and adding jobs in $\mathcal{K}$ in some arbitrary way. Then, $\Sc'$ corresponds to a solution obtained by adding jobs from $\mathcal{K}$ with a list-scheduling procedure (for some tie-breaking rule) if and only if the number of jobs in $\mathcal{K}$ added to each machine $i$ is: (i) $\left\lceil \tfrac{(\lambda-\ell_i(\Sc))_+}{p}\right\rceil$ if $\tfrac{(\lambda-\ell_i(\Sc))_+}{p}$ is not an integer, and either $\tfrac{(\lambda-\ell_i(\Sc))_+}{p}$ or $\tfrac{(\lambda-\ell_i(\Sc))_+}{p}+1$ if $\tfrac{(\lambda-\ell_i(\Sc))_+}{p}$ is a non-negative integer.
\end{observation}

Our main procedure is called every time that we get a new job $j^*$ (where $\J' = \J \cup\{j^*\})$ and receives as input the current solution $\Sc$ for $(\J,\M)$. If $\J=\emptyset$, then $\Sc$ is trivially initialized as empty. The exact description is given in Algorithm~\ref{LPTonline}. 

Broadly speaking, the algorithm works in phases $h\in\{0,\ldots,|\tP|\}$, where for each $h$ it assigns jobs in $\J'_{h}$. First, we assign jobs exactly as in $\Sc_{h}$ for machines in which the assignment of $\Sc_{h-1}$ and $\Sc'_{h-1}$ coincide. The set of such machines is denoted by $\Me_{h-1}$ and the set of remaining machines is denoted by $\Mne_{h-1}$. As we will see, this is consistent with LPT by the previous observation and Lemma~\ref{load}. The remaining jobs in $\J'_h$ are assigned using list-scheduling. Crucially, we will break ties in favor of machines where the assignment of $\Sc_{h-1}$ and $\Sc'_{h-1}$ differ. This is necessary to avoid creating new machines with different assignments. After assigning huge and big jobs, small jobs are added exactly as in $\Sc$ in machines where the assignment of big jobs in $\Sc$ and $\Sc'$ coincides. The rest of small jobs are added greedily. In the last part, the algorithm rebalances small jobs by moving them from machines of load higher than $\ell_i(\Sc')+2^{\ell}$ to the least loaded machines.

\begin{algorithm}
	\caption{Online LPT}
	\label{LPTonline}
	\algorithmicrequire{ \parbox[t]{12.75cm}{ 
			Instances $(\J,\M)$ and $(\J',\M)$ such that $\J'=\J\cup\{j^*\}$; a schedule $\Sc(\J,\M)$.}}
	
	\begin{algorithmic}[1]
		\State{run LPT on input $\J'$ and let $\tau$ be the minimum load of the constructed solution. Set $\UB\gets 2\tau$. Define $\tP, \ell$, and $u$ based on this upper bound $\UB$ using \eqref{eq:deftP1} and \eqref{eq:deftP2}.}
		\State{set $\Me_{-1}\gets\M$ and $\Mne_{-1}\gets\emptyset$.}
		
		\For{$h=0,1,\ldots,|\tP|$} \Comment{Assignment of big and huge jobs}
		\State{for each machine $i\in \Me_{h-1}$, assign all jobs in $\J_h\cap \Sc^{-1}(i)$ to $i$ in $\Sc'$.}\label{st:bigMe}
		\State{for jobs in $\J'_h$ still not assigned in $\Sc'$, apply list-scheduling (with an arbitrary order of jobs). If there is more than one least loaded machine break ties in favor of machines in $\Mne_{h-1}$.}
		\label{st:bigMne}
		\State{define $\Me_h$ as the set of machines $i$ such that $\Sc_h^{-1}(i)=\Sc_h'^{-1}(i)$ and $\Mne_h\gets\M \setminus \Me_h$.}
		\EndFor			
		\For{machines $i\in \Me_{|\tP|}$}\Comment{Assignment of small jobs} 
		\State{assign all small jobs w.r.t to $\UB$ in $\J\cap \Sc^{-1}(i)$ to $i$ in $\Sc'$.}
		\EndFor
		
		\State{\label{st:listSmall}assign the remaining jobs using list-scheduling}.
		\State{set $\overline{\M}$ to be the set of machines containing a small job w.r.t $\UB$.} 
		
		\While{there exists $i\in \overline{\M}$ s.t. $\ell_i(\Sc')> \ell_{\min}(\Sc')+2^{\ell}$}\label{st:reassignSmall}
		\State{consider a machine $i\in \overline{\M}$ of maximum load. Reassign the smallest job in $\Sc'^{-1}(i)$ to any least loaded machine.}
		\State{update $\overline{\M}$ to be the set of machines containing a small job w.r.t $\UB$.}
		\EndWhile
		\State \Return $\Sc'$.
	\end{algorithmic}
\end{algorithm} 

We can bound the competitive ratio of the algorithm in a very similar way to Lemma~\ref{lm:1.7-competitive}. First we prove the following auxiliary lemma.

\begin{lemma}\label{lm:LPT-solution} If $\Sc'$ is the output of the algorithm then $\Sc'_{|\tP|}$ is an LPT-solution. \end{lemma}

\begin{proof} We show the proof inductively. Consider a run of the algorithm with input assignment~$\Sc$. If $\Sc$ is empty then it is clearly an LPT-solution. Otherwise, $\Sc$ is the output of a run of the algorithm. 
	We can assume inductively that $\Sc_{|\tP_0|}$ is an LPT-solution (and thus also any restriction of $\Sc_{|\tP_0|}$ to jobs of sizes at least $p$, for any $p\ge0$). Notice that $\UB_0\le \UB$, by Lemma~\ref{load}, and thus $\min \tP_0\le \min \tP$ and $\max \tP_0\le \max \tP$.
	
	We use a second induction to show that, for every $h\in\{0,\ldots,|\tP|\}$, $\Sc_h'$ is an LPT-solution.  
	
	To show the base case ($h=0$), consider jobs in $J_0'$, which are all larger than $\UB\ge \OPT'$. Hence there are at most $m$ of them, and the algorithm assigns them each to a different machine (this, again, follows inductively). Thus, the base case holds.	 
	
	
	Consider $h\ge 1$ and let us assume that $\Sc_{h-1}'$ is an LPT-solution.  Let $\Sc_{\text{LPT},h}$ be an LPT-solution for jobs in $\J_0\cup \ldots \cup \J_{h}$, and similarly $\Sc'_{\text{LPT},h}$ for jobs in $\J'_0\cup \ldots \cup \J'_{h}$. First observe that the load profile vector  $\load(\Sc'_{\text{LPT},h})$ is independent of the tie-breaking rule. Consider the target value $\lambda = \ell_{\min}(\Sc_{\text{LPT},h})$ and $\lambda'=\ell_{\min}(\Sc'_{\text{LPT},h})$. Notice that, by Lemma~\ref{load}, $\lambda\le \lambda'$.
	
	Since $\Sc_{h-1}$ is an LPT-solution, then $\Sc'_h$ is an LPT-solution if jobs in $\J'_h$ are added using list-scheduling. By Observation~\ref{obs:LScharac} the following characterizes this fact: for all machines $i\in M$, the number of jobs assigned in $\Sc'_h$ to $i$ is: (i) $\lceil (\lambda'-\ell_i(\Sc'_{h-1}))_+/q_h\rceil$ if $(\lambda'-\ell_i(\Sc'_{h-1}))_+/q_h$ is not an integer, and either $(\lambda'-\ell_i(\Sc'_{h-1}))_+/q_h$ or $(\lambda'-\ell_i(\Sc'_{h-1}))_+/q_h+1$ if $(\lambda'-\ell_i(\Sc'_{h-1}))_+/q_h$ is an integer. Since $\lambda \le \lambda'$, and $\Sc_h$ is an LPT-solution, then the number of jobs assigned in Step~\ref{st:bigMe} is never more than $\lceil (\lambda'-\ell_i(\Sc'_{h-1}))_+/q_h\rceil$ if $(\lambda'-\ell_i(\Sc'_{h-1}))_+/q_h$ is not an integer, and never more than $(\lambda'-\ell_i(\Sc'_{h-1}))_+/q_h+1$ if $(\lambda'-\ell_i(\Sc'_{h-1}))_+/q_h$ is integer. This implies that after adding jobs in Step~\ref{st:bigMne} we obtain an LPT-solution. \end{proof}

Now we can argue about the approximation guarantee of the obtained solution.

\begin{lemma}\label{lm:4/3-competitive}
	When considering instances of Machine Covering such that $|\M|=m$, Algorithm~\ref{LPTonline} is $(\frac{4m-2}{3m-1}+O(\e))$-competitive. 
\end{lemma}

\begin{proof}
	We will use the previous lemma to show that $\Sc'$ is a $(k,k)$-relaxed version of $\Sc_{\LPT}(\J',\M)$ for some $k\le 4$, which is enough to conclude the claim due to Lemma~\ref{RobGree} and the result from Csirik et al.~\cite[Theorem~$3.5$]{CKW92}. Indeed, let $k = 2^{\ell}/(\e\OPT')$. Then, by the previous lemma all jobs larger than $k\e\OPT'=2^{\ell}$ are assigned with LPT. Also, the while loop at Step~\ref{st:reassignSmall} ensures that the output of the algorithm is a $(k,k)$-relaxed version of $\Sc_{\LPT}(\J',\M)$. The lemma follows since $k\le 4$ as shown in Lemma~\ref{lm:1.7-competitive}.
\end{proof}

\paragraph*{\textbf{Bounding the migration factor.}} To analyze the migration factor of the algorithm, we will show that $|\Mne_{|\tP|}|$ is upper bounded by a constant. This will be done inductively by first bounding $|\Mne_{h}\setminus \Mne_{h-1}|$ for each $h$ and then using the fact that $|\tP|\in O((1/\e)\log(1/\e))$. A description of the overall idea can be found in Figure~\ref{fig:migration}.

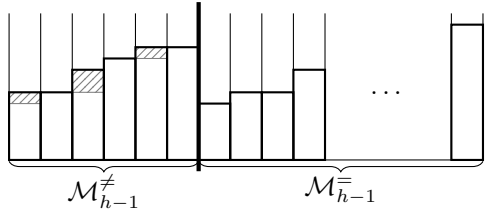
\begin{figure}
	\centering
\begin{tikzpicture}[yscale=0.3, xscale=0.7]


\draw[fill=gray, color=gray, pattern=north east lines, pattern color=gray] (0,2.5) rectangle (0.6,3);
\draw[fill=gray, color=gray, pattern=north east lines, pattern color=gray] (1.2,3) rectangle (1.8,4);
\draw[fill=gray, color=gray, pattern=north east lines, pattern color=gray] (2.4,4.5) rectangle (3,5);


\draw[thick] (0,0) rectangle (0.6,3);
\draw[thick] (0.6,0) rectangle (1.2,3);
\draw[thick] (1.2,0) rectangle (1.8,4);
\draw[thick] (1.8,0) rectangle (2.4,4.5);
\draw[thick] (2.4,0) rectangle (3,5);
\draw[thick] (3,0) rectangle (3.6,5);
\draw[thick] (3.6,0) rectangle (4.2,2.5);
\draw[thick] (4.2,0) rectangle (4.8,3);
\draw[thick] (4.8,0) rectangle (5.4,3);
\draw[thick] (5.4,0) rectangle (6,4);
\draw[thick] (8.4,0) rectangle (9,6);


\draw (0,6.5) -- (0,0) -- (9,0) -- (9,6.5);
\draw (0.6,0) -- (0.6,6.5);
\draw (1.2,0) -- (1.2,6.5);
\draw (1.8,0) -- (1.8,6.5);
\draw (2.4,0) -- (2.4,6.5);
\draw (3,0) -- (3,6.5);
\draw (3.6,0) -- (3.6,6.5);
\draw (4.2,0) -- (4.2,6.5);
\draw (4.8,0) -- (4.8,6.5);
\draw (5.4,0) -- (5.4,6.5);
\draw (6,0) -- (6,6.5);
\draw (8.4,0) -- (8.4,6.5);


\draw[ultra thick] (3.6,-0.5) -- (3.6,7);
\draw [decorate,decoration={brace,amplitude=4pt,raise=0.5pt}]
(3.6,0) -- (0,0);
\draw (1.8,-0.2) node[anchor=north] {\small $\Mne_{h-1}$};
\draw [decorate,decoration={brace,amplitude=4pt,raise=0.5pt}]
(9,0) -- (3.6,0);
\draw (6.3,-0.3) node[anchor=north] {\small $\Me_{h-1}$};
\draw (7.2,3) node {\small $\dots$};

\end{tikzpicture}
	\caption{\footnotesize{Depiction of a possible situation at the end of iteration $h-1$. The machines on the right side correspond to machines in $\Me_{h-1}$ and therefore process the same jobs in $\Sc_{h-1}$ and $\Sc'_{h-1}$. 
			Assume, possibly erroneously and just as a thought experiment, that the machines in $\Mne_{h-1}$ can be sorted non-decreasingly by load for $\Sc_{h-1}$ and $\Sc'_{h-1}$ simultaneously. The two solutions are depicted simultaneously in the picture, where the difference of loads on machines in $\Mne_{h-1}$ corresponds to the dashed area. The total dashed load equals to $\tp_{j^*}$, which is spread in only constantly many machines by Lemma~\ref{PocaSub}. When assigning jobs in $\J_h$, the algorithm first assigns a number of jobs to each machine in $\Me_{h-1}$ (Step~\ref{st:bigMe}), and then fills machines in $\Mne_{h-1}$. Notice that while the algorithm does not assign another job to a machine in $\Me_{h-1}$, no new machine will enter $\Mne_{h}\setminus \Mne_{h-1}$. On the other hand, the number of such jobs can be bounded by a number proportional to $\tp_{j^*}$ (and $1/\e$), which then also bounds the number of machines in $\Mne_{h}\setminus \Mne_{h-1}$. In reality, however, it is not true that the machines in $\Mne_{h-1}$ can be sorted non-decreasingly on the loads for $\Sc_{h-1}$ and $\Sc'_{h-1}$ simultaneously. This provokes a number of technical difficulties that we avoid by using a different permutation of machines for each solution and invoking Lemma~\ref{load}.}}\label{fig:migration}
\end{figure}

Let us consider huge jobs w.r.t $\UB$ (i.e. jobs in $\J'_0$). Notice that all these jobs are larger than $\OPT'\ge \OPT$, and hence in $\Sc'_{0}$ each one is assigned alone to one machine. The same situation happens in solution $\Sc$ restricted to jobs in $\J_0$. Thus, none of these jobs are migrated. Hence, we can assume w.l.o.g. for the sake of the analysis of the migration that all jobs are big or small w.r.t $\UB$ (including $j^*$). Additionally, we can assume that $j^*$ is not small, since otherwise there is no migration.

Let $\Je_h$ be the set of jobs assigned by Step~\ref{st:bigMne} to machines in $\Me_{h-1}$. Notice that the jobs in $\Je_h$ correspond to the jobs in $\J'_{h}$ that $\Sc'$ assigns to a machine in $\Me_{h-1}$ but $\Sc$ processes in $\Mne_{h-1}$. Our strategy will be to bound the cardinality of set $\Je_h$ and then use this to upper bound $|\Mne_h \setminus \Mne_{h-1}|$. First we prove two auxiliary lemmas that help to upper bound $|\Je_h|$.

\begin{lemma}
	\label{lm:LowerBoundAssignmentSingleMachine}
	Assume that $\Je_h\neq \emptyset$. For each machine $i\in \Mne_{h-1}$, if $\lambda-\ell_i(\Sc_{h-1}')\ge 0$ solution $\Sc'_{h}$ assigns to $i$ at least $\left\lfloor\tfrac{(\lambda-\ell_i(\Sc_{h-1}'))_+}{q_h}\right\rfloor +1 $ many jobs from $\J_h$.
\end{lemma}   
\begin{proof} 	 
	We consider two cases. If $\lambda'=\lambda$, then the number of jobs in $\J_h$ assigned to machine $i$ is at least $\left\lfloor\tfrac{(\lambda'-\ell_i(\Sc_{h-1}'))_+}{q_h}\right\rfloor +1$. Indeed, if $\tfrac{(\lambda'-\ell_i(\Sc_{h-1}'))_+}{q_h}$ is fractional then the number of jobs must be $\left\lceil\tfrac{(\lambda'-\ell_i(\Sc_{h-1}'))_+}{q_h}\right\rceil = \left\lfloor\tfrac{(\lambda'-\ell_i(\Sc_{h-1}'))_+}{q_h}\right\rfloor +1$. If, on the other hand, $\tfrac{(\lambda'-\ell_i(\Sc_{h-1}'))_+}{q_h}$ is integral, then the algorithm might assign to $i\in \Mne_{h-1}$ only $\tfrac{(\lambda'-\ell_i(\Sc_{h-1}'))_+}{q_h}=\left\lfloor\tfrac{(\lambda-\ell_i(\Sc_{h-1}'))_+}{q_h}\right\rfloor$ many jobs. However, if this is the case the tie-breaking rule in Step~\ref{st:bigMne} implies that $\Je_h= \emptyset$, which contradicts our hypothesis. Then the number of assigned jobs is exactly $\left\lfloor\tfrac{(\lambda'-\ell_i(\Sc_{h-1}'))_+}{q_h}\right\rfloor +1$, and thus the claim holds.
	
	If $\lambda'>\lambda$, then $\lambda'>\ell_i(\Sc_{h-1}') $ and the number of jobs in $\J_h$ assigned to machine $i$ is at least $\left\lceil\tfrac{(\lambda'-\ell_i(\Sc_{h-1}'))_+}{q_h}\right\rceil = \left\lceil\tfrac{\lambda'-\ell_i(\Sc_{h-1}')}{q_h}\right\rceil\ge \left\lfloor\tfrac{(\lambda-\ell_i(\Sc_{h-1}'))}{q_h}\right\rfloor +1=\left\lfloor\tfrac{(\lambda-\ell_i(\Sc_{h-1}'))_+}{q_h}\right\rfloor +1$ and hence the claim holds. 	 
\end{proof}

\begin{lemma}\label{lm:LoadRemoveEntry} Let $x,y\in \mathbb{R}_+^n$, $x=(x_1,x_2,\dots,x_n)$, $y=(y_1,y_2,\dots,y_n)$ such that $x_1\le x_2 \le \dots\le x_n$, $y_1\le y_2\le \dots\le y_n$ and $x\le y$ coordinate-wise. Assume that $x_j=y_i$ for some indices $i,j$. If $x_{-i}$ denotes the $(n-1)$-dimensional vector obtained by removing the $i$-th entry of $x$, and $y_{-j}$ is the vector obtained by removing the $j$-th entry of $y$, then $x_{-i}\le y_{-j}$.\end{lemma}

\begin{proof} Notice first that if $i=j$ then the result is a direct consequence of Lemma~\ref{loadprof_generalized}: by taking $\alpha=\beta=-x_i$ and coordinate $i$, we get new vectors $\tilde{x}$ and $\tilde{y}$ satisfying $\tilde{x} \le \tilde{y}$ and $\tilde{x}_1 = \tilde{y}_1 = 0$, and hence we can conclude that $x_{-i} \le y_{-j}$ because $x_{-i}$ (resp $y_{-j}$) corresponds to the last $n-1$ coordinates of $\tilde{x}$ (resp. $\tilde{y}$).
	
	We now distinguish two cases: if $i<j$, we have that $y_j = x_i \le x_j \le y_j$, hence $x_k = y_j$ for every $k=i,i+1,\dots,j$. This implies that $x_{-i} = x_{-(i+1)} = \dots = x_{-j}$, and then we can conclude that $x_{-i} \le y_{-j}$ by applying the previous observation for $x_{-j}$ and $y_{-j}$. On the other hand, if $j<i$, we define vector $z$ equal to $y$ but replacing coordinates $j, j+1, \dots, i$ by $y_j$. It is not difficult to see that $x \le z \le y$ coordinate-wise, and also $z_{-j} = z_{-(j+1)} = \dots = z_{-i}$. If we apply the first observation for $x$ and $z$ using coordinate $i$ we have that $x_{-i}\le z_{-i}$, and applying it to $z$ and $y$ using coordinate $j$ we get that $z_{-j} \le y_{-j}$. Merging both inequalities and using the fact that $z_{-i} = z_{-j}$, we conclude that $x_{-i} \le y_{-j}$. \end{proof}

\begin{lemma}\label{lm:Jeq}
	It holds that $|\Je_h|\in O( \frac{\tilde{p}_{j^*}}{\e2^{\ell}})$.
\end{lemma}
\begin{proof}
	Assume, w.l.o.g., that $\Mne_{h-1}=\{1,\ldots,m'\}$ and that $\ell_1(\Sc'_{h-1})\le \ell_2(\Sc'_{h-1}) \le \ldots \le \ell_{m'}(\Sc'_{h-1})$. Consider also a permutation $\sigma:\Mne_{h-1}\rightarrow \Mne_{h-1}$ such that $\ell_{\sigma(1)}(\Sc_{h-1})\le \ell_{\sigma(2)}(\Sc_{h-1}) \le \ldots \le \ell_{\sigma(m')}(\Sc_{h-1})$. By Lemma~\ref{load} the sorted vector of loads (over all machines) of solution $\Sc_{h-1}$ is upper bounded by the sorted vector of loads of $\Sc'_{h-1}$. Applying Lemma~\ref{lm:LoadRemoveEntry} iteratively to remove machines in $\Me_{h-1}$ one by one (which have the same assignment in both solutions), it holds that $\ell_{\sigma(i)}(\Sc_{h-1}) \le \ell_{i}(\Sc'_{h-1})$ for all $i\in \Mne_{h-1}$. Let us consider sets \begin{eqnarray*} T_{-}&=&\{i\in \Mne_{h-1}\,:\, \ell_i(\Sc'_{h-1})\le\lambda\}, \text{ and }\\ T_{+}&=&\{i\in \Mne_{h-1}\,:\, \ell_{\sigma(i)}(\Sc_{h-1})\le \lambda \text{ and } \ell_i(\Sc'_{h-1})>\lambda\}.	 \end{eqnarray*}
	Lemma~\ref{lm:LowerBoundAssignmentSingleMachine} implies that the total number of jobs from $\J_h'$ assigned by $\Sc'_{h}$ to machines in $\Mne_{h-1}$ is at least
	\begin{eqnarray*}
		& & \sum_{i\in T_-} \left(\left\lfloor \tfrac{(\lambda - \ell_i(\Sc'_{h-1}))_+}{q_h} \right\rfloor+1\right) \\
		& = &\sum_{i\in T_-} \left(\left\lfloor \tfrac{(\lambda - \ell_{\sigma(i)
			}(\Sc_{h-1}))_+}{q_h} \right\rfloor+1\right) +\sum_{i\in T_-} \left\lfloor \tfrac{(\lambda - \ell_i(\Sc'_{h-1}))_+}{q_h} \right\rfloor -\left\lfloor \tfrac{(\lambda - \ell_{\sigma(i)}(\Sc_{h-1}))_+}{q_h} \right\rfloor\\
		& = &\sum_{i\in T_-\cup T_+} \left(\left\lfloor \tfrac{(\lambda - \ell_{\sigma(i)
			}(\Sc_{h-1}))_+}{q_h} \right\rfloor+1\right) - \sum_{i\in T_+} \left(\left\lfloor \tfrac{(\lambda - \ell_{\sigma(i)
			}(\Sc_{h-1}))_+}{q_h} \right\rfloor+1\right)\\
		& &+ \sum_{i\in T_-} \left\lfloor \tfrac{(\lambda - \ell_i(\Sc'_{h-1}))_+}{q_h} \right\rfloor -\left\lfloor \tfrac{(\lambda - \ell_{\sigma(i)}(\Sc_{h-1}))_+}{q_h} \right\rfloor.	 
	\end{eqnarray*}
	Notice that the set $T_-\cup T_+$ contains all indices $i\in\Mne_{h-1}$ such that $\ell_{\sigma(i)}(\Sc_{h-1})\le \lambda$. Hence, the first sum in the last expression upper bounds the number of jobs in $\J_{h}$ that solution $\Sc$ assigns to machines in $\Mne_{h-1}$. That way, since $|\J'_h\setminus \J_h | \le 1$, it holds that 
	
	\begin{eqnarray*}
		|\Je_h| &\le& 1 + \sum_{i\in T_+} \left(\left\lfloor \tfrac{(\lambda - \ell_{\sigma(i)
			}(\Sc_{h-1}))_+}{q_h} \right\rfloor+1\right)
		+ \sum_{i\in T_-} \left\lfloor \tfrac{(\lambda - \ell_{\sigma(i)}(\Sc_{h-1}))_+}{q_h} \right\rfloor-\left\lfloor \tfrac{(\lambda - \ell_i(\Sc'_{h-1}))_+}{q_h} \right\rfloor \\
		&\le& 1 + |T_+|+\sum_{i\in T_+} \left\lfloor \tfrac{(\lambda - \ell_{\sigma(i)
			}(\Sc_{h-1}))_+}{q_h} \right\rfloor
		+ \sum_{i\in T_-} \left\lfloor \tfrac{(\lambda - \ell_{\sigma(i)}(\Sc_{h-1}))_+}{q_h} \right\rfloor-\left\lfloor \tfrac{(\lambda - \ell_i(\Sc'_{h-1}))_+}{q_h} \right\rfloor \\
		&\le& 1 + |T_+|+\underbrace{\sum_{i\in T_+} \left\lfloor \tfrac{(\lambda - \ell_{i
				}(\Sc'_{h-1}))_+}{q_h} \right\rfloor}_{=0}
		+ \sum_{i\in T_-\cup T_+} \left\lfloor \tfrac{(\lambda - \ell_{\sigma(i)}(\Sc_{h-1}))_+}{q_h} \right\rfloor-\left\lfloor \tfrac{(\lambda - \ell_i(\Sc'_{h-1}))_+}{q_h} \right\rfloor. 
	\end{eqnarray*}
	
	Let us now consider $T_{\neq}=\{i\in \Mne_{h-1}: \ell_{\sigma(i)}(\Sc_{h-1})\neq \ell_{i}(\Sc'_{h-1})\}$. Thus, the last expression is at most 
	\begin{eqnarray*}
		|\Je_h| & \le& 1 + |T_+| + \sum_{i\in (T_-\cup T_+)\cap T_{\neq}} \left\lfloor \tfrac{(\lambda - \ell_{\sigma(i)}(\Sc_{h-1}))_+}{q_h} \right\rfloor-\left\lfloor \tfrac{(\lambda - \ell_i(\Sc'_{h-1}))_+}{q_h} \right\rfloor \\	 
		&\le& 1 + |T_+| + \sum_{i\in T_{\neq}} \left( \tfrac{(\lambda - \ell_{\sigma(i)}(\Sc_{h-1}))_+}{q_h} - \tfrac{(\lambda - \ell_i(\Sc'_{h-1}))_+}{q_h} +1 \right)\\
		&\le& 1 + |T_+| + |T_{\neq}| + \sum_{i\in T_{\neq}} \tfrac{\ell_i(\Sc'_{h-1})- \ell_{\sigma(i)}(\Sc_{h-1})}{q_h}\\
		&\le& 1 + 2|T_{\neq}| + \tfrac{\tilde{p}_{j^*}}{q_h}.
	\end{eqnarray*}	 
	Also, Lemma~\ref{PocaSub} can be applied and thus $|T_{\neq}|\le \frac{\tilde{p}_{j^*}}{\e2^{\ell}}$. 
	The lemma finally follows since $q_h\ge 2^{\ell}$ by definition.
\end{proof}

Notice that jobs in $\J_h^{=}$ are the only jobs assigned in a given iteration $h$ that can cause one new machine to have different assignments in $\Sc_{h}$ and $\Sc'_h$. Thus, $|\Mne_{h}\setminus \Mne_{h-1}|\le |\J_h^{=}|$ and the following lemma holds.

\begin{lemma}\label{lm:Mne_variation}
	For all $h\in\{1,\ldots,|\tP|\}$ it holds that $|\Mne_{h}\setminus \Mne_{h-1}|\in O( \frac{\tilde{p}_{j^*}}{\e2^{\ell}})$.
\end{lemma}

Putting all the discussed ideas together, we prove the following result.

\begin{theorem}\label{thm:4/3+mig}
	When considering instances of Machine Covering such that $|\M|=m$, Algorithm Online LPT is a polynomial time $(\frac{4m-2}{3m-1}+O(\e))$-competitive algorithm with $O((1/\e^3)\log(1/\e))$ migration factor.
\end{theorem}

\begin{proof} 	We first argue that the algorithm runs in polynomial time. Indeed, it suffices to show that the algorithm enters the while loop in Step~\ref{st:reassignSmall} a polynomial number of times. This follows easily as the quantity $\ell_{\min}(\Sc')$ is non-decreasing, and hence a job can be reassigned to a least loaded machine at most once. Notice that the competitive ratio of the algorithm follows from Lemma~\ref{lm:4/3-competitive}.
	
	Let us now bound the migration factor. We do this in two steps. First consider solution~$\Sc'$ before entering Step~\ref{st:reassignSmall}. We first bound the volume of jobs migrated between $\Sc$ and $\Sc'$, and then bound the total volume of jobs reassigned in the while loop in Step~\ref{st:reassignSmall}. 
	
	For the first bound, by the previous lemma and since $\Mne_{-1}=\emptyset$, it holds that $|\Mne_{|\tP|}|\le |\tP|\cdot O( \frac{\tilde{p}_{j^*}}{\e2^{\ell}})\le O((1/\e)\log(1/\e) \frac{\tilde{p}_{j^*}}{\e2^{\ell}})$. The load of jobs in $\Sc'_{|\tP|}$ that are migrated is upper bounded by $\sum_{i\in\Mne_{|\tP|}} \ell_i(\Sc_{|\tP|}')\le|\Mne_{|\tP|}|\max_{i\in \Mne_{|\tP|}}\ell_i(\Sc_{|\tP|}')$. On the other hand, since we are assuming (w.l.o.g) that there is no huge job, the total load of each machine is at most $2\UB$ as argued in Section~\ref{sec:Rounding}. We conclude that the big jobs migrated have a total load of at most $2\UB\cdot|\Mne_{|\tP|}| =\UB\cdot O((1/\e)\log(1/\e) \frac{\tilde{p}_{j^*}}{\e2^{\ell}}).$ Finally, notice that small jobs migrated (before entering Step~\ref{st:reassignSmall}) are the ones assigned to machines in $\Mne_{|\tP|}$ by $\Sc$. Since $\Sc$ is the output of Online LPT, then the total load of these jobs is at most $(\ell_{\min}(\Sc')+2^{\ell})\cdot \Mne_{|\tP|}\le 2\UB\cdot \Mne_{|\tP|}\le\UB\cdot O((1/\e)\log(1/\e) \frac{\tilde{p}_{j^*}}{\e2^{\ell}})$. We conclude that the total load migrated is at most $\UB\cdot O((1/\e)\log(1/\e) \frac{\tilde{p}_{j^*}}{\e2^{\ell}})$.
	
	It remains to bound the volume migrated in the while loop of Step~\ref{st:reassignSmall}. For this we will show the following claim.
	
	\noindent\textbf{Claim:} Let $\Sc'$ be the solution constructed before  entering Step~\ref{st:reassignSmall}. Then all reassigned jobs in the while loop, except possibly the one reassigned last, are assigned to a machine in $\Mne_{|\tP|}$ by $\Sc'$.
	
	Assume the claim holds and let us consider the solution $\Sc'$ as output by the algorithm. Then the total volume of reassigned jobs is bounded by $|\Mne_{|\tP|}|\max_{i\in \Mne_{|\tP|}}\ell_i(\Sc')$. Since by construction the load of a machine that process a job smaller than $2^{\ell}$ is at most $\ell_{\min}(\Sc')+2^{\ell}\le 2\UB$, the total volume migrated will be at most $\UB\cdot O((1/\e)\log(1/\e) \frac{\tilde{p}_{j^*}}{\e2^{\ell}})$ as before. Hence, the migration factor is upper bounded by 
	\[
	O\left(\frac{\UB}{\tilde{p}_{j^*}}\cdot(1/\e)\log(1/\e) \frac{\tilde{p}_{j^*}}{\e2^{\ell}}\right) = O((1/\e^3)\log(1/\e)).	
	\]
	
	To show the claim, consider $\Sc'$ before entering Step~\ref{st:listSmall} together with the corresponding set $\overline{\M}$ of machines that process some small job. Since $\Sc$ is the output of Online LPT, then the difference between the maximum an minimum loads of machines in $\overline{\M}\cap\Me_{|\tP|}$ for solution $\Sc'$ is at most $2^{\ell}$. We call this property (P1). Also, notice that $\overline{\M}\cap \Mne_{|\tP|} =\emptyset$, and hence, the maximum load difference of two machines in this set is at most $2^{\ell}$, vacuously. We refer to this property as (P2). 
	Notice that (P1) and (P2) 
	hold iteratively throughout the later steps of the algorithm. Additionally, if some job is assigned to a machine $\Me_{|\tP|}$ in Step~\ref{st:listSmall}, the algorithm does not enter the while loop and we are done. Otherwise, the minimum load is achieved at $\Mne_{|\tP|}$. Hence, if there is a job migrated from a machine in $\Me_{|\tP|}$ to $\Mne_{|\tP|}$ then the algorithm finishes. The claim follows. \end{proof}

\subsection{A note on geometric v/s arithmetic rounding}\label{sec:geom_round}

One of the main reasons to use our rounding procedure to multiples of $\e2^\ell$ instead of the geometric rounding (i.e., down to the nearest power of $(1+\e)$) is because the same arguments used in this work cannot be applied to geometric rounded instances. It is crucial in the analysis that the number of possible loads is $\text{poly}\left(1/\e\right)$, while for geometric rounded instances that is not true as the following lemmas show.

\begin{lemma}\label{DifLoa} Let $\e \in \mathbb{Q}_+, \e < 1$. Given a machine covering instance $(\J,\M)$, let $\tilde{\J}$ be the set of jobs obtained by rounding geometrically jobs with processing time $p_j \in [\e\OPT,\OPT]$. If $C_1, C_2 \subseteq \tilde{J}$ are two different multi-sets of jobs with processing times at least $\e \OPT$ such that $\sum_{j\in C_i}{p_j} \in [\e\OPT,\OPT]$, $i=1,2$, then $\sum_{j\in C_1}{p_j} \neq \sum_{j\in C_2}{p_j}$. \end{lemma}

\begin{proof} Assume w.l.o.g. $\OPT=1$. Hence the possible processing times are $(1+\e)^i$, with $i$ such that $\e \le (1+\e)^i \le 1$ (a finite family of such possible values). Suppose by contradiction that there are two different non-empty multi-sets $C_1, C_2$ with the same total load, and assume they are minimal, i.e. that there is no other pair of non-empty multi-sets with the same total load but with smaller total load. For $k=1,2$, let $C_k(j)$ be the number of jobs with processing time $(1+\e)^j$ in set $C_k$. Since the pair $C_1$, $C_2$ is minimal, we have that $C_1(j) = 0$ or $C_2(j) = 0$ for every $j$. $C_1$ and $C_2$ having the same total load means that \begin{equation*} \displaystyle\sum_{j=-k}^{0}{C_1(j) (1+\e)^j} = \displaystyle\sum_{j=-k}^{0}{C_2(j) (1+\e)^j}, \end{equation*} where $k = -\lfloor \log_{1+\e}(\e)\rfloor$. This last equality can be rephrased as the existence of a non-zero polynomial $p(x) = b_0 + b_1 x + \dots + b_k x^k$, with $\lvert b_j \rvert \in \{C_1(j),C_2(j)\}$ (i.e. with integer coefficients), that has $(1+\e)$ as one of its roots. Since $\e= \frac{c}{d} > 0$ for some co-primes $c$ and $d$, then $1+\e= \frac{c+d}{d}$. Dividing $p(x)$ by $(dx-(c+d))$ leads to a polynomial $q(x)=a_0 + a_1 x + \dots + a_{k-1}x^{k-1}$ which, thanks to Gauss lemma, has integer coefficients too. Let $b_i$ be the first coefficient of $p$ different from zero. Then \begin{equation*} \lvert b_i \rvert = \lvert (c+d) a_i + d a_{i-1}\rvert= \lvert\left(c + \frac{c}{\e}\right)a_i + \frac{c}{\e}a_{i-1} \rvert > \frac{1}{\e}, \end{equation*} implying, since the size of each job is at least $\e$, that the total load of the multi-sets is at least $b_i \e>1$, which is a contradiction. \end{proof}

\begin{lemma}\label{ExpConf} Given $0<\e<1$, the number of different multi-sets of jobs with processing time at least $\e\OPT$ with total load at most $\OPT$ for a geometrically rounded instance is $2^{\Omega\left(\frac{1}{\e}\right)}$. \end{lemma}

\begin{proof} Let $u=\lfloor \log_2 \OPT \rfloor$ and $\ell = \lceil \log_2 (\e\OPT)\rceil$. We will give a lower bound on the number of different sets with total load $2^u$ when the jobs are rounded to powers of $2$, which implies that for $0<\e<1$ the same bound holds for processing times rounded to powers of $(1+\e)$. Let $\mathcal{C}_i$ be the number of different multi-sets with total load $2^{\ell + i}$. This number is characterized by the recurrence \begin{align*} C_0 & = 1 \\ C_{i+1} & = 1 + \frac{C_i(C_i+1)}{2}. \end{align*} This last term comes from the fact that a multi-set with total load $2^{\ell+i+1}$ can be constructed using only one job of size $2^{\ell+i+1}$, or merging two multi-sets of size $2^{\ell+i}$ (there are $\binom{C_i}{2} + C_i = \frac{C_i(C_i+1)}{2}$ such pairs). 
	
	Since recurrence $a_0=1$, $a_i = \frac{a_i^2}{2}$ satisfies $a_i \ge 2^{2^{i}}$, we conclude that $C_{u-\ell} \ge 2^{2^{\log\frac{1}{\e}}} \ge 2^{\Omega\left(\frac{1}{\e}\right)}$. \end{proof}

Because of these two lemmas, if we use geometrically rounded instances we cannot make sure that, when a new jobs arrives to the system, the load profile changes only by $\text{poly}\left(1/\e\right)$ coordinates since there are $2^{\Omega\left(1/\e\right)}$ number of possible different loads.

\subsection{An improved lower bound for the competitive ratio with constant migration factor}\label{sec:lowerbound}

In opposition to online makespan minimization with migration, where competitive ratio arbitrarily close to one can be achieved using a constant migration factor \cite{SSS09}, the online machine covering problem does not allow it. Until now, the best lower bound known for this ratio is $\frac{20}{19}$ \cite{SV10}, which we now improve to $\frac{17}{16}$ using similar ideas.

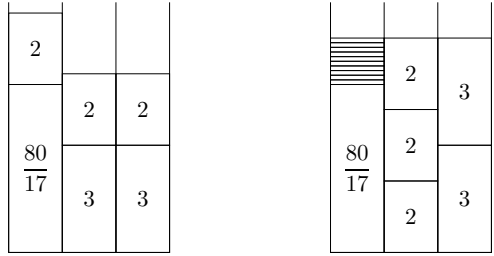
\begin{figure}
	\centering
	\resizebox{!}{95pt}{
\begin{tikzpicture}[scale=0.6]


\draw (2.5,7) -- (2.5,0) -- (-2,0) -- (-2,7);
\draw (1,7) -- (1,0);
\draw (-0.5,7) -- (-0.5,0);


\draw  (-0.5,0) rectangle (1,3);
\draw (0.25,1.5) node {$3$};
\draw  (1,0) rectangle (2.5,3);
\draw (1.75,1.5) node {$3$};
\draw  (-0.5,3) rectangle (1,5);
\draw (0.25,4) node {$2$};
\draw  (1,3) rectangle (2.5,5);
\draw (1.75,4) node {$2$};
\draw  (-2,4.7) rectangle (-0.5,6.7);
\draw (-1.25,5.7) node {$2$};
\draw  (-2,0) rectangle (-0.5,4.7);
\draw (-1.25,2.35) node {$\dfrac{80}{17}$};


\draw (11.5,7) -- (11.5,0) -- (7,0) -- (7,7);
\draw (10,7) -- (10,0);
\draw (8.5,7) -- (8.5,0);


\draw  (10,0) rectangle (11.5,3);
\draw (10.75,1.5) node {$3$};
\draw  (10,3) rectangle (11.5,6);
\draw (10.75,4.5) node {$3$};
\draw  (8.5,0) rectangle (10,2);
\draw (9.25,1) node {$2$};
\draw  (8.5,2) rectangle (10,4);
\draw (9.25,3) node {$2$};
\draw  (8.5,4) rectangle (10,6);
\draw (9.25,5) node {$2$};
\draw  (7,0) rectangle (8.5,4.7);
\draw (7.75,2.35) node {$\dfrac{80}{17}$};


\draw (7,4.7) rectangle (8.5,4.83);
\draw (7,4.83) rectangle (8.5,4.96);
\draw (7,4.96) rectangle (8.5,5.09);
\draw (7,5.09) rectangle (8.5,5.22);
\draw (7,5.22) rectangle (8.5,5.35);
\draw (7,5.35) rectangle (8.5,5.48);
\draw (7,5.48) rectangle (8.5,5.61);
\draw (7,5.61) rectangle (8.5,5.74);
\draw (7,5.74) rectangle (8.5,5.87);
\draw (7,5.87) rectangle (8.5,6);

\end{tikzpicture}}
	\caption{Left: Unique $(17/16)$-approximate solution before the arrival of small jobs. Right: Unique $(17/16)$-approximate solution after the arrival of small jobs.}
	\label{cotainf}
\end{figure} 

\begin{lemma}\label{17/16} For any $\e>0$, there is no $\left(\frac{17}{16}-\e\right)$-competitive algorithm using constant migration factor for the online machine covering problem with migration. \end{lemma}

\begin{proof} Consider an instance consisting of $3$ machines and $6$ jobs of sizes $p_1=p_2=p_3=2$, $p_4=p_5=3$ and $p_6=\frac{80}{17}$. It is easy to see that the optimal solution is given by Figure \ref{cotainf} (a). Moreover, there is no other $\left( \frac{17}{16}-\e \right)$-approximate solution (up to symmetry).
	
	Suppose by contradiction that there exists a $\left( \frac{17}{16}-\e\right)$-competitive algorithm with constant migration factor $C$. While processing the above instance, the algorithm must construct the optimal solution depicted in Figure \ref{cotainf} (left). Consider now that jobs with processing time smaller than $1/C$ arrive to the system, with total processing time $\frac{22}{17}$. Since the migration factor is $C$, none of the six previous jobs can be migrated, thus the best minimum load we can obtain is $\frac{96}{17}$, while the optimal solution is $6$ (see Figure \ref{cotainf} (right)). We conclude by noting that $\frac{6}{96/17}=\frac{17}{16}$. \end{proof}

Notice that the instance reaching the lower bound crucially depends on the arrival of jobs with arbitrarily small processing times. This kind of jobs are in fact the problematic ones, because under the assumption that at each iteration the incoming job is big enough (has processing time at least $\e\OPT$), there is a robust PTAS with constant migration factor~\cite{SV10}.

\bibliographystyle{plain}
\bibliography{bibliografia}{}

\begin{thebibliography}{10}

\bibitem{AE98}
Y.~Azar and L.~Epstein.
\newblock On-line machine covering.
\newblock {\em J. Sched.}, 1:67--77, 1998.

\bibitem{BJK15}
S.~Berndt, K.~Jansen, and K.~Klein.
\newblock Fully dynamic bin packing revisited.
\newblock In {\em {APPROX/RANDOM} 2015}, pages 135--151.

\bibitem{CEKvS13}
Xujin Chen, Leah Epstein, Elena Kleiman, and Rob van Stee.
\newblock Maximizing the minimum load: The cost of selfishness.
\newblock {\em Theor. Comput. Sci.}, 482:9 -- 19, 2013.

\bibitem{CKW92}
J.~Csirik, H.~Kellerer, and G.~Woeginger.
\newblock The exact {LPT}-bound for maximizing the minimum completion time.
\newblock {\em Oper. Res. Lett.}, 11:281--287, 1992.

\bibitem{DFL82}
B.~Deuermeyer, D.~Friesen, and M.~Langston.
\newblock Scheduling to maximize the minimum processor finish time in a
  multiprocessor system.
\newblock {\em SIJADM}, 3:190--196, 1982.

\bibitem{EL09}
L.~Epstein and A.~Levin.
\newblock A robust {APTAS} for the classical bin packing problem.
\newblock {\em Math. Program.}, 119:33--49, 2009.

\bibitem{EL14}
L.~Epstein and A.~Levin.
\newblock Robust algorithms for preemptive scheduling.
\newblock {\em Algorithmica}, 69:26--57, 2014.

\bibitem{FNS04}
A.~Frangioni, E.~Necciari, and M.~Scutellà.
\newblock A multi-exchange neighborhood for minimum makespan parallel machine
  scheduling problems.
\newblock {\em J. Comb. Optim.}, 8:195--220, 2004.

\bibitem{GGK16}
A.~Gu, A.~Gupta, and A.~Kumar.
\newblock The power of deferral: Maintaining a constant-competitive steiner
  tree online.
\newblock {\em {SIAM} J. Comput.}, 45:1--28, 2016.

\bibitem{HS88}
D.~Hochbaum and D.~Shmoys.
\newblock A polynomial approximation scheme for scheduling on uniform
  processors: Using the dual approximation approach.
\newblock {\em {SIAM} J. Comput.}, 17:539--551, 1988.

\bibitem{JK13}
K.~Jansen and K.~Klein.
\newblock A robust {AFPTAS} for online bin packing with polynomial migration.
\newblock In {\em {ICALP} 2013}, pages 589--600.

\bibitem{JKV16}
K.~Jansen, K.~Klein, and J.~Verschae.
\newblock Closing the gap for makespan scheduling via sparsification
  techniques.
\newblock In {\em {ICALP} 2016}, pages 1--13.

\bibitem{Lacki2015}
J.~\L{acki}, J.~O\'{c}wieja, M.~Pilipczuk, P.~Sankowski, and A.~Zych.
\newblock The power of dynamic distance oracles: Efficient dynamic algorithms
  for the steiner tree.
\newblock In {\em {STOC} 2015}, pages 11--20.

\bibitem{MSVW16}
N.~Megow, M.~Skutella, J.~Verschae, and A.~Wiese.
\newblock The power of recourse for online {MST} and {TSP}.
\newblock {\em {SIAM} J. Comput.}, 45:859--880, 2016.

\bibitem{RRSV10}
D.~Recalde, C.~Rutten, P.~Schuurman, and T.~Vredeveld.
\newblock {Local Search Performance Guarantees for Restricted Related Parallel
  Machine Scheduling}.
\newblock {\em {LATIN} 2010}, pages 108--119, 2010.

\bibitem{SSS09}
P.~Sanders, N.~Sivadasan, and M.~Skutella.
\newblock Online scheduling with bounded migration.
\newblock {\em Math. Oper. Res.}, 34:481--498, 2009.

\bibitem{SV07}
P.~Schuurman and T.~Vredeveld.
\newblock Performance guarantees of local search for multiprocessor scheduling.
\newblock {\em {INFORMS} J. Comput.}, 19:52--63, 2007.

\bibitem{SV10}
M.~Skutella and J.~Verschae.
\newblock Robust polynomial-time approximation schemes for parallel machine
  scheduling with job arrivals and departures.
\newblock {\em Math. Oper. Res.}, 41:991--1021, 2016.

\bibitem{V07}
B.~V\"ocking.
\newblock Selfish load balancing.
\newblock In {\em Algorithmic Game Theory}, pages 517--542. 2007.

\bibitem{W97}
G.~Woeginger.
\newblock A polynomial-time approximation scheme for maximizing the minimum
  machine completion time.
\newblock {\em Oper. Res. Lett.}, 20:149--154, 1997.

\end{thebibliography}

\newpage
\appendix
\section{Non-constant Migration for classic LPT}
\label{app:nonConstant}

\begin{lemma}\label{lm:nonConstantMigration} For any $k\ge 2$ there exists a set $\J$ of $4k+1$ jobs and an extra job $j^*\notin \J$ such that, for every schedule $\Sc$ constructed using $\LPT$ on $2k+1$ machines, it is not possible to construct a schedule $\Sc'$ using $\LPT$ for $\J \cup \{j^*\}$ with migration factor less than $m/2$. \end{lemma}

\begin{proof} Fix a constant $0<\e\le \frac{1}{6k}$. Consider a set $\J$ consisting of the following $4k+1$ jobs: $k+1$ jobs of size $1$; for each $i\in \{0,\dots,k-1\}$, a job of size $\frac{1}{2}+i\e$ and a job of size $\frac{1}{2}-(i+1)\e$, and finally $k$ jobs of size $\frac{1}{2}-k\e\ge\frac{1}{3}$. Assume the jobs in $\J$ are sorted non-increasingly by size. There is a unique schedule constructed using $\LPT$ for this instance (up to symmetry) which assigns the jobs in the following way (see Figure \ref{fig:LPT_no_red1}): The $k+1$ jobs of size $1$ to a machine on their own, and for each $i=1,\dots,k$, it assigns to machine $k+i$ a job of size $\frac{1}{2} + (k-i-1)\e$, a job of size $\frac{1}{2}-(k-i)\e$ and a job of size $\frac{1}{2}-k\e$ (since the total load of the first two jobs is $1-\e$, the last $k$ jobs must be assigned to these $k$ machines). 
	
Now consider an arriving job $j^*$ of size $\frac{1}{2} + k\e\le\frac{2}{3}$. There is a unique schedule constructed using $\LPT$ for the new instance (up to symmetry) which assigns the jobs in the following way (see Figure \ref{fig:LPT_no_red2}): it assigns to the first $k+1$ machines a job of size $1$ and a job of size $\frac{1}{2}-k\e$, to machine $k+2$ job $p_{j^*}$ and a job of size $\frac{1}{2}-(k-1)\e$, for each $i=2,\dots,k-1$ it assigns to machine $k+i+1$ a job of size $\frac{1}{2} + (k+1-i)\e$ and a job of size $\frac{1}{2} - (k-i)\e$, and finally to machine $2k+1$ a job of size $\frac{1}{2}+\e$ and a job of size $\frac{1}{2}$ (now the total load of machines $k+2,\dots,2k+1$ is $1+\e$, then the last $k+1$ jobs must be assigned to the first $k+1$ jobs).
	
It is not difficult to see that, in the new schedule, every machine has a different subset of jobs assigned to it compared with the original schedule, and so at least one job must have been migrated per machine. Thus, the migrated total load is at least the load of the smallest $2k+1$ jobs, which implies that the needed migration factor is at least \begin{eqnarray*} \frac{\displaystyle\sum_{i=0}^{2k}{p_{(4k+1)-i}}}{p_{j^*}} \ge\frac{(2k+1)\left(\frac{1}{2}-k\e\right)}{\frac{1}{2}+k\e} \ge m\frac{1/3}{2/3} = \frac{m}{2}. \qedhere \end{eqnarray*} \end{proof}

\section{\textbf{Proof of Theorem~\ref{thm:1.7apx}.}}\label{sec:chenetal} We will use and generalize a result from Chen et al.~\cite{CEKvS13} which bounds the price of anarchy of a related game. We say that a schedule is \emph{lex-jump-optimal} if the solution is locally optimal with respect to Jump but considering the whole vector of loads (sorted non-decreasingly) as weight function and comparing them lexicographically. In this context, lex-jump-optimal solutions are equivalent to pure Nash equilibria for the Machine Covering game, obtained if jobs are selfish agents trying to minimize the load of the machine where they are assigned and the minimum load is considered as the welfare function.

\begin{theorem}[Chen et al.~\cite{CEKvS13}]\label{thm:chenetal} The Price of Anarchy of the Machine Covering Game is $1.7$. \end{theorem}

Theorem~\ref{thm:chenetal} gives a tight bound for the approximation ratio of lex-jump-optimality. In order to prove Theorem~\ref{thm:1.7apx} we need to prove that in the case of Machine Covering a jump-optimal solution is also lex-jump-optimal, and to construct instances to prove the lower bound for swap-optimality.

\begin{lemma}\label{lem:lextojump} Let $\Sc$ be a jump-optimal solution for Machine Covering. Then $\Sc$ is lex-jump-optimal as well. \end{lemma}

\begin{proof} Suppose $\Sc$ is not lex-jump-optimal. This implies that there is a job $j$ and two machines $i,i'$, where $j$ is assigned to $i$, such that $\ell_i(\Sc)-p_j>\ell_{i'}(\Sc)$. Since $\ell_{i'}(\Sc)\ge \ell_{\min}(\Sc)$, we get that $\ell_i(\Sc)-p_j>\ell_{\min}(\Sc)$, which implies that $\Sc$ is not jump-optimal thanks to Lemma~\ref{CarOL}. \end{proof}

\begin{lemma}\label{lem:swapLB} The approximation ratio of swap-optimality is at least $1.7$. \end{lemma}

\begin{proof} The family of instances leading to the desired lower bound is a slight modification of the one presented by Chen et al.~\cite{CEKvS13} in the proof of Theorem~\ref{thm:chenetal}, because the original family of instances is jump-optimal but not swap-optimal.
	
	Let $n_0=0$ and, for each $i\ge 1$, $n_i=4n_{i-1}+2$. For each $k\ge 2$ we will define an instance of Machine Covering $(\J_k,\M_k)$. Let $\delta=\frac{1}{30n_k}$ and $|\M_k|=2(10^k-1)$. There will be five types of jobs:
	
	\begin{itemize}
		\item $|\M_k|$ jobs of size $1$;
		\item For each $i=1,\dots,k$, we create $6\cdot 10^{k-i}$ jobs of size $a_i=\frac{1}{2}+(n_i-1)\delta$;
		\item For each $i=1,\dots,k$, we create $12\cdot 10^{k-i}$ jobs of size $b_i=\frac{1}{2}-(n_i-1)\delta$;
		\item For each $i=1,\dots,k$, we create $12\cdot 10^{k-i}$ jobs of size $c_i=\frac{1}{5}+4n_{i-1}\delta$;
		\item For each $i=1,\dots,k-1$, we create $6\cdot 10^{k-i}$ jobs of size $d_i=\frac{1}{5}-n_i\delta$, and finally $6|\M_k|$ jobs of size $d_k=\frac{1}{6|\M_k|-1}$.
	\end{itemize} 
	
	Notice first that the optimal solution for instance $(\J_k,\M_k)$ achieves a minimum load of at least $1.7-\delta$: we can assign a job of size $1$ to each machine, and on top of that either a job of size $a_i$ plus a job of size $d_i$ (for some $1\le i \le k-1$), or a job of size $b_i$ plus a job of size $c_i$ (for some $1\le i \le k$), or a job of size $a_k$ plus $|\M_k|$ jobs of size $d_k$. Since $1+a_i+d_i=1.7-\delta$ for $1\le i\le k$, $1+b_i+c_i=1.7+(4n_{i-1}-n_i+1)\delta = 1.7 - \delta$ and $1+a_k+\frac{|\M_k|}{6|\M_k|-1}=1.7-\delta+\frac{1}{6(6|\M_k|-1)} \ge 1.7-\delta$, we get the desired bound. The number of machines is $18\displaystyle\sum_{i=1}^{k}{10^{k-i}}=2(10^k-1)$ as required.
	
	Consider now the solution $\Sc_{swap}$ which assigns the jobs in the following way:
	
	\begin{itemize}
		\item Jobs of size $1$ are assigned in pairs to $\frac{|\M_k|}{2}$ machines, being the load of such machines $2$;
		\item For each $i=1,\dots,k$, we assign one job of size $a_i$ and two jobs of size $b_i$ to $6\cdot 10^{k-i}$ machines, being the load of such machines $\frac{3}{2}-(n_i-1)\delta$;
		\item We assign six jobs of size $c_1=\frac{1}{5}$ to $2\cdot 10^{k-1}$ machines, being the load of such machines $\frac{6}{5}$;
		\item For each $i=1,\dots,k-1$, we assign one job of size $c_{i-1}$ and five jobs of size $d_i$ to $12\cdot 10^{k-i-1}$ machines, being the load of such machines $\frac{6}{5}-n_i\delta$;
		\item All the jobs of size $d_k$ goes to a final machine which will have load $1+\frac{1}{6|\M_k|-1}$.
	\end{itemize}
	
	The number of machines is $$\frac{|\M_k|}{2}+6\displaystyle\sum_{i=1}^{k}{10^{k-i}}+2\cdot 10^{k-1} + 12\displaystyle\sum_{i=1}^{k-1}{10^{k-i-1}}+1=\frac{|\M_k|}{2}+10^k-1=|\M_k|,$$ so the solution is feasible. Furthermore, the last machine is the only least loaded machine as the smallest load among the remaining machines is $\frac{6}{5}-n_k\delta=\frac{7}{6}$.
	
	We will now prove that $\Sc_{swap}$ is swap-optimal. Since for each job $j$ we have that $\ell_{\Sc_{swap}^{-1}(j)}(\Sc_{swap})-p_j\le 1$ and the minimum load is at least $1$, $\Sc_{swap}$ is jump-optimal. Consider now any job $j$ assigned to a machine $i$ which is not the least loaded one and a job $j'$ of size $d_k$ assigned to the least loaded machine $i'$. If job $j$ is the smallest job assigned to $i$ then swapping it with $j'$ does not improve the solution because $\ell_{i}(\Sc)-p_j=\ell_{i'}(\Sc_{swap})-p_{j'}=1$. On the other hand, if $p_j$ is strictly larger than the rest of the processing times of jobs in $i$ (i.e. $p_j=c_q$ or $a_q$ for some $q=1,\dots,k$), then $\ell_{i}(\Sc)-p_j<1$, implying that $\ell_{i}(\Sc)-p_j+p_{j'}<\ell_{i'}(\Sc_{swap})$, and hence not improving the solution. This proves that $\Sc_{swap}$ is swap-optimal.
	
	By taking $k$ increasing, $\delta$ decreases and approaches zero, implying that the approximation ratio of swap-optimality is at least $1.7$. \end{proof}

By putting together Theorem~\ref{thm:chenetal} and Lemma~\ref{lem:swapLB} we can conclude the proof of Theorem~\ref{thm:1.7apx}. \qed
	
\end{document}